\documentclass[smallextended]{svjour3}       
\smartqed  
\usepackage{color,caption}
\usepackage{amsmath}
\usepackage{amsmath, amssymb, graphicx}
\usepackage[tight,footnotesize]{subfigure}
\usepackage{textcomp}
\usepackage{mathtools}
\usepackage{cite}
\usepackage{csquotes}

\newcommand{\EditSZ}[1]{{\color{blue} #1}}
\newcommand{\CommentSZ}[1]{}

\def\beq{\begin{equation}}
\def\eeq{\end{equation}}

\begin{document}

\title{Model Predictive Control under Timing Constraints induced by Controller Area Networks
\thanks{The research work is supported by NSF grants ECCS-0841195 (CAREER), CNS-0931576, and CMMI-1436284. Authors would like to
thank Prof. Haibo Zeng for the discussion of CAN bus.}
}


\author{ Zhenwu Shi   \and Fumin Zhang
}


\institute{Zhenwu Shi and Fumin Zhang\at
              Georgia Institute of Technology
              \email{zwshi, fumin@gatech.edu}           
}

\date{Received: date / Accepted: date}

\maketitle

\begin{abstract}
\noindent When multiple model predictive controllers  are implemented on a shared  controller area network (CAN),   their performance may degrade due to the variable timing and delays among messages. The priority based real-time scheduling of messages on the CAN introduces complex timing of events, especially when the types and number of messages change at runtime. This paper introduces a novel hybrid timing model to make runtime predictions on the timing of the messages for a finite time window.
Controllers can be designed using the optimization algorithms for model predictive control by considering  the timing as optimization constraints.  This timing model allows multiple controllers to share a CAN without significant degradation in the controller performance. The timing model also provides a convenient way to check the schedulability of messages on the CAN at runtime. Simulation results demonstrate that the timing model is accurate and computationally efficient to meet the needs of real-time implementation. Simulation results also demonstrate that model predictive controllers designed when considering the timing constraints have superior performance than the controllers designed without considering the timing constraints.
\end{abstract}

%
%
%
%
%
%
%
%
%
%
\section{Introduction} \label{section:Introduction}
Modern industrial control applications, such as the automotive control, are characterized by  the use of shared networks to replace excessive wiring.
Deterministic timing is crucial in  time-critical industrial applications, because uncertainty in timing may cause embarrassing, or even life-threatening, sudden decreases in systems performance. Real-time networks have been developed to support networking with deterministic timing, with the control area network, or CAN, being the most mature and accepted one \cite{CAN2.0, Zeng2010}. A CAN connects a number of nodes that are able to send and receive messages. Each message on the CAN  is broadcasted to all nodes, and only one message can be transmitted at any time. To resolve contention among multiple messages, the CAN utilizes a media access control protocol  called carrier sense multiple access with bitwise arbitration (or CSMA/BA). Each message is assigned a unique identifier, which is used as an assigned priority when contention occurs.   Since each identifier is unique, each message has a unique priority. Therefore, when two or more nodes attempt to send messages at the same time, the node with the highest priority message will be granted access to the CAN to transmit, and the other nodes will need to defer their message transmission until the communication link becomes idle, which can be detected after receiving a bit field indicating the end of the message being transmitted.  The length of each CAN message can be determined up to certain accuracy and uncertainties so that the
value well approximate the real timing and there is no randomness in the mechanism to resolve conflicts. Therefore, the timings of message transmission and reception events on CAN can be well predicted.

Using  CAN to support networked control systems increases flexibility. However,  most
networked control system designs are usually  constrained by limited bandwidth of the communication link, which does not allow message transmission at an arbitrarily high
rate.  The CAN based control systems are no exception. When multiple control loops must share access to a common communication link, the bandwidth must be distributed appropriately so that all control loops are stable and all achieve a desired level of performance. Hence, one must design both the controller and the distribution of bandwidth to guarantee stability and  optimal and robust performance \cite{Hespanha2007,Zhang2013}.

Over the last several decades, hardware and software systems supporting the CAN have  improved significantly, resulting in very reliable message transmission and timing accuracy. Therefore, the probability of packet drops, and the possible randomness in timing caused by clock drift, can be  practically ignored for controller design. However,  since the CSMA/BA used by CAN is a contention based  protocol, it  alone cannot provide sufficient control over the distribution of bandwidth among networked controllers.  While the timing is still deterministic, contention may cause large  variations in timing, a phenomenon generally known as jitters  \cite{Baruah1997,Cervin2003}. If not handled well,  jitters may cause controls to be faulty at unexpected  (or  even  life threatening) times. These timing variations cannot be ignored by any controller design. But some jitters happen with small probability, and so are  hard  to diagnose \cite{Cervin2006}.
Since the contention based  media access protocol is not sufficient to avoid timing variations,  a higher level scheduling algorithm is often designed to allocate the bandwidth among control loops. In \cite{CANTabuada}, authors discuss the design of self-triggered controllers that can  reduce the number of required messages for control systems, which can save  communication bandwidth for other applications. Also,  authors of \cite{Marti2010} propose an optimal strategy to allocate communication bandwidth to different control loops implemented on a CAN, and the article \cite{Jeon2001} analyzes the effect of response time on the control performance. However, these methods cannot completely avoid contention.
When unexpected contention occurs, classical real-time scheduling resolves  these contentions by priority based scheduling algorithms \cite{ShaRev04}, such as  the popular rate monotonic scheduling (or RMS) and earliest deadline first (or EDF) algorithms \cite{LiuLayland73}.

Model predictive controllers, or MPCs, were originally  developed for industrial process control \cite{Clarke1987, Lee1994,Richalet1978}. The success led to a new general approach for
controller design that has been used in many other applications, including vehicle and robot control \cite{ Camacho2004,Grune2011,Wang2009}.  The basic idea of
model predictive control is to  use a model of the physical system to predict future system behavior over a finite time horizon, starting from each sampling time where
new sensor measurements are available. The control effort in the finite time horizon is computed by solving an optimization problem. At each sampling time, only the first value of the resulting control is applied to the plant, and then the entire calculation is  repeated at subsequent sampling time instants. Model predictive control offers a natural way to incorporate state and control constraints \cite{Mayne2000} to the design. However, it requires sufficient online computing resources and computing time. Chemical process control, where model predictive control has seen great success over many years, allows for both. While other applications, such as the control of automobiles or robots, are more constrained in terms of  timing and computing resources because of their reliance on (networked) embedded computers \cite{leen2002}, recent advances in embedded processors show considerable promise for applying model predictive control in automotive and robotic applications as well.

Effective MPC designs rely on accurate, high fidelity  models of the control loops. However, the jitters associated with messages on the CAN incur time-varying delays into the control loops. Such time-varying delays make it difficult to derive a reliable model used by MPC. This challenge may be answered by the approach of {\it control-scheduling co-design} where a controller and the timing of control related events can be jointly determined. Two categories of methods exist in the literature: the offline methods and the online methods.

Offline methods perform optimization at an offline design stage \cite{HuLemmon06, ZWS_RTSS08}. Typically, a scheduling algorithm is first determined, and then computer simulation is used to generate a sequence of timings of all possible events under the scheduling algorithm.  Optimization techniques can then be applied to tune the parameters of the scheduling algorithm and the controller design until a predefined performance criteria is optimized \cite{Arzen2000}. Offline methods are feasible. However, they lead to overly conservative designs, and they are not completely compatible with model predictive control that requires control efforts to be computed online for a finite time window using predictions. Online methods for handling jitters involve co-designing a schedule of events and a model predictive control at each sampling time for a finite horizon \cite{Henriksson2002}, which reduce the amount of computation required compared to offline methods since a shorter time window is concerned.  Furthermore, the controllers that are designed in such methods are usually  less conservative than the ones designed with  offline methods, because they only need to compensate for  the worst case delay in a relatively short time window. A key requirement for the online approach is a computationally efficient method to predict the timing of events for the finite horizon used by the model predictive control.  While timing can be computed by simulations for  offline methods, such simulations are too expensive for online methods.  To the best of our knowledge,  the existing works do not offer a general method for accurate timing prediction on real-time networks.  For example, the works \cite{Gaid2006,Liu2013} obtain a timing prediction from a lookup table that is generated offline by computer simulation. In \cite{Zhao2008},  the timing is assumed to be periodic, while \cite{Zhang2005} models the timing as a Markov process, where the transition probabilities are assumed to be known. These methods all have certain  degrees of inaccuracy that must be tolerated by a model predictive control design.  If a message takes longer than expected to transmit, or is perturbed by other messages that were not considered at the design stage so that  its deadline is missed, then the schedule would not adjust for this fault.
The work \cite{Shi2013}  is perhaps the first to introduce a deterministic timing model that connects real-time priority based scheduling algorithms with model predictive control designs. This timing model may be leveraged by model predictive control designs to improve performance, by better compensating for timing variations, which serves as the starting point of the work of this paper.

This paper develops a novel methodology that focuses on handling the timing constraints (e.g. jitters and delays) associated with MPC on CAN. The major contributions are summarized as below:
\begin{itemize}
\item {\bf Model.} We develop a receding horizon timing model for  event-triggered model predictive control on  CAN.  Existing real-time scheduling analysis of the CAN  focuses on modeling time-varying delays as either constant values in worst-case scenarios\cite{Tindell,Tindell1995, CanRevisit} or stochastic variables obeying certain distribution\cite{Zeng2010}. These results do not provide a process model with sufficient accuracy. Moreover, many control systems nowadays are operating in dynamic and uncertain environments. As a result, the system workload will change accordingly. For instance, some messages on the CAN may need to be removed in some cases, while new messages may be added in other cases. This variability of messages inevitably further increases delay variation in feedback control loop.  Our model is able to capture the variations of timing caused by the changes of number of messages, message length, and  priorities at run-time over a finite time window. This is particularly suitable for model predictive control.
\item {\bf MPC design.}  \EditSZ{We propose an effective design for an event-triggered MPC that incorporates both the timing model and the control loop model to find the optimal controlling effort under the timing constraints on CAN.}   Networked model predictive control designs exist for  contention based  protocols over the Ethernet; see, for example,   \cite{Goodwin2004,Imer2006,Liu2007,Loontang2006,Montestruque2004}. However, the Ethernet is very different from the CAN bus, since it does not offer predictable timing. Therefore, these works cannot be applied directly to the model predictive control design problem for the CAN bus. Our MPC design is triggered by the deterministic timing events on the CAN. We have discovered that a state observer is necessary to estimate the states of the timing model. An observer with proved convergence is thus incorporated into the MPC design. The observer and the event triggered MPC controller design have not appeared in previous works.
\item {\bf Simulations.} We perform simulations to demonstrate  that our MPC design can lead to improved MPC performance. The design is compared to MPC designs without the timing model to show the performance improvement.
\end{itemize}
To the best of our knowledge, these contributions do not exist in the literatures reviewed and have not been previously published.

The technical content of the paper is organized as follows.  Section \ref{section:Problem} first review the CAN protocol and its message properties. Then a structure for MPC is introduced, which formulates the co-design problem studied in this paper. This problem  motivates the need for an efficient timing model that is necessary to enable the co-design.  Section \ref{section:TimingModel} then derives the timing model that is needed to solve the co-design problem. The  timing model consists state vectors, selected to represent the status of all messages on the CAN bus, and transition rules that
determine the values of  state vectors over time. Using the timing model, one can check for schedulability of all messages at significant moments. Not all states in the state vectors are directly observable on a CAN bus, Section \ref{section:StateObserver} discusses how to estimate the state vectors in the hybrid timing model from measurements collected in each CAN node. We rigorously prove that the algorithm used for estimation converges to the true values of the state vectors.  Section \ref{section:MPC} presents the solution of MPC design proposed in this paper. The timing model is used to determine controller delays so that the MPC can be determined more accurately than using worst case response times. Section \ref{section:simulation} presents simulations to show the effectiveness of our approach. We demonstrate that the timing model is at least as accurate as other simulation based methods, but significantly reduced computational cost. We also demonstrate that the co-designed MPC with timing model achieves better tolerance to disturbances in timing than using worst-case timing. For ease of reading, we have summarized all major notations used throughout the paper in Table 1.

\begin{table}[hbtp]
\label{allnotation}
\centering
\vspace{-2mm}
\caption{Major Notations in Paper}
\begin{tabular}{cl} \hline \vspace{-2mm}\\
 & \hspace{1.5cm}CAN Bus Messages \vspace{1mm}\\
\hline \vspace{-2mm}\\
$\tau_n$ & message chain consisting of sensor and control messages \vspace{1mm}\\
$\tau_n^1$ & 1st sub-message of $\tau_n$, i.e. the sensor message \vspace{1mm}\\
$\tau_n^2$ & 2nd sub-message of $\tau_n$, i.e. the control message \vspace{1mm}\\
$C_n^1$ & transmission duration of $\tau_n^{1}$ \vspace{1mm}\\
$C_n^2$ & transmission duration of $\tau_n^{2}$\vspace{1mm}\\
$I_n^1$ & time for preparing $\tau_n^{1}$ \vspace{1mm}\\
$I_n^2$ & time for preparing $\tau_n^{2}$ \vspace{1mm}\\
$T_n$ &  sampling interval of $\tau_n$ \vspace{1mm}\\
$P_n$ & priority of $\tau_n$ \vspace{1mm}\\
$\alpha_n$ & sampling instant of $\tau_n$ \vspace{1mm}\\
$\beta_n$ & time instant when $\tau_n^{1}$ finishes transmission \vspace{1mm}\\
$\gamma_n$ & time instant when $\tau_n^{2}$ finishes transmission  \vspace{1mm}\\
$\delta_n$ & time delay between $\gamma_n$ and $\alpha_n$ \vspace{1mm}\\
\hline \vspace{-2mm}\\
 & \hspace{1.5cm}MPC Control Design \vspace{1mm}\\
\hline \vspace{-2mm}\\
$x$ & state variable of a physical plant \vspace{1mm}\\
$y$ & output of a physical plant \vspace{1mm}\\
$u$ & MPC control signal applied on the physical plant \vspace{1mm}\\
$J$ & cost function for MPC design \vspace{1mm}\\
$T_p$ & length of MPC prediction horizton \vspace{1mm}\\
$\lambda$ & reference trajectory that MPC tracks \vspace{1mm}\\
\hline \vspace{-2mm}\\
 & \hspace{1.5cm}Timing Model \vspace{1mm}\\
\hline \vspace{-2mm}\\
$N$ & number of total message chains on the CAN bus \vspace{1mm}\\
$d_n$ & deadline state of a message chain $\tau_n$ \vspace{1mm}\\
$r_n$ & residue state of a message chain $\tau_n$ \vspace{1mm}\\
$o_n$ & delay state of a message chain $\tau_n$ \vspace{1mm}\\
$D$ & deadline state of all message chains, i.e. $D=[d_1,\dotsb, d_N]$ \vspace{1mm}\\
$R$ & residue state of all message chains, $R=[r_1,\dotsb, r_N]$ \vspace{1mm}\\
$O$ & delay state of all message chains, $O=[o_1,\dotsb, o_N]$ \vspace{1mm}\\
$I\hspace{-0.55mm}D$ &  index of the message chain being transmitted on CAN \vspace{1mm}\\
$Z$ &  state vector of the model, i.e. $Z=[D, R, O, I\hspace{-0.55mm}D]$\vspace{1mm}\\
$\mathbb{H}$ &  the timing model \vspace{1mm}\\
\hline
\end{tabular}
\end{table}

\section{Problem Formulation} \label{section:Problem}
Our main goal is to establish an event-triggered model predictive control design approach for real-time networks. An ``event" is defined as a significant moment  that should be accounted for by the controller. For example, each time a sensor message finishes transmission, a model predictive  controller can be  initiated to leverage the new information.  Event-triggered model predictive control fits nicely with  the CAN bus, since the CAN hardware can  generate hardware interrupts when ``end of transmission" events  happen. We propose a timing model so that whenever the model predictive control is triggered by an event, one can predict the timing of future events within a finite time horizon and compute control effort accordingly.  For example, one can predict when a future sensor message will arrive and  when the corresponding control effort will be applied, and then compute the control effort for that future time.

 Without loss of generality, we make the following technical assumptions about message transmission and reception on the CAN bus:
\begin{enumerate}
\item The CAN bus is reliable such that no error occurs in sending and receiving messages.
\item At each node, among all messages that are ready for transmission, the message with the highest priority will be sent first.
\end{enumerate}
These two assumptions are valid in real applications, and have been used in many theoretical works related to CAN \cite{Tindell1995,CanRevisit,CANTabuada}.

\begin{figure}[htp]
\centering
\includegraphics[width=0.7\textwidth]{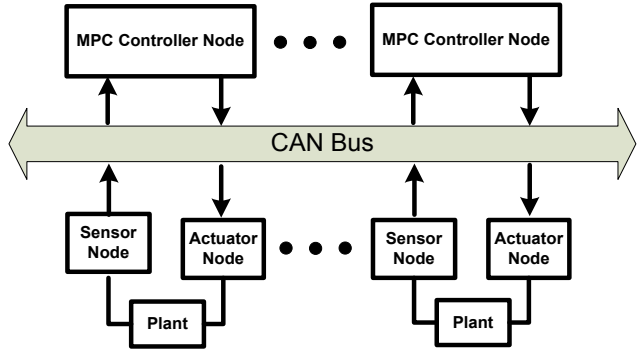}
\caption{Multiple Feedback Control Loops Sharing a CAN} \label{fig:network}
\end{figure}

\subsection{CAN-based Control System} \label{section:Problemsetup}
Consider a set of feedback control loops designed to share a CAN  as illustrated in Figure \ref{fig:network}. Each feedback control loop utilizes the CAN  to send sampled data from sensors to an MPC controller, and to send control commands from the MPC controller to actuators. The sensors, MPC controllers, and actuators are connected to the CAN and are named as sensor nodes, MPC controller nodes, and actuator nodes. We simplify the design so that each feedback control loop has one sensor node, one MPC controller node, and one actuator node. This is not to be considered as only allowing single-input-single-output systems because multiple sensors can be integrated into a sensor node, and multiple actuators can be integrated into an actuator node.  The following rules are imposed by this system:
\begin{itemize}
\item At the sensor node, a user specified software program samples the state of the plants, and then combines sampled data into a single sensor message  for transmission;
\item At the MPC controller node, upon reception of a sensor message, a user specified software program extracts sampled data from the sensor message. The node then computes MPC algorithms, and combines the resulting control commands into a single control message for transmission;
\item At the actuator node, upon reception of a control message, a user specified software program extracts control commands from the control message. The node then issues the control on the actuator;
\item  All control loops are mutually independent, which means that the sensor messages and control messages of  one control loop do not rely on messages from other loops for computation.
\end{itemize}
Therefore, we consider two types of messages related to the control: sensor messages and actuator messages.
The above rules imply a {\bf causality constraint}  between sensor and control messages as follows:
in each feedback control loop, a sensor message must be transmitted {\bf before} the MPC controller starts computing the control law. A control message can only be transmitted {\bf after} the control law is computed.

\subsection{Message Chains}
For causality in each feedback control loop, one requires that the transmission of  a sensor message  be followed by the computation of the control effort, which is then followed by the transmission of a control message.
This process iteratively repeats. Each iteration of this process, beginning from the sampling of sensor and ending at the actuation, is called an {\it instance}, and then the above process for any $n$-th feedback control loop is called a {\it message chain} and denoted by $\tau_n$. Thus, each message chain $\tau_n$ is composed of recurring instances. Let the indices $k=1,2,...$  indicate each of the recurring instances in $\tau_n$ for the $n$-th loop i.e.  the $k$-th instance of $\tau_n$ is  denoted by $\tau_n[k]$.   Figure \ref{fig:NCtask}  illustrates the timing of a message chain when there is no contention.

\begin{figure}[hbtp]
\centering
\includegraphics[width=0.9\textwidth]{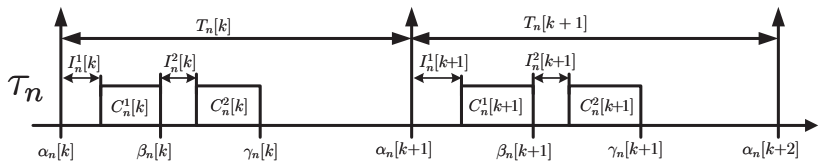}
\caption{An Example Message Chain $\tau_n$  when No Contention Occurs} \label{fig:NCtask}
\end{figure}

The horizontal line in Figure \ref{fig:NCtask}  represents the progression of time. Suppose that $\tau_n[k]$ starts at the $k$-th sampling instant $\alpha_n[k]$. The instance $\tau_n[k]$ contains two sub-messages, namely, $\tau_n^{1}[k]$ and $\tau_n^{2}[k]$, where  $\tau_n^{1}[k]$ represents the sensor message,  and $\tau_n^{2}[k]$ represents the control message. Also,  $I_n^{1}[k]$ is the amount of time for the user specified software program on the sensor node to sample plants and prepare $\tau_n^{1}[k]$; $C_n^1[k]$ is the transmission duration of $\tau_n^{1}[k]$; $\beta_n[k]$ is the time instant when the transmission of $\tau_n^1[k]$ is completed; $I_n^{2}[k]$ is the amount of time for  the user specified software program on the controller node to extract sensor information, compute the model predictive control, and prepare $\tau_n^{2}[k]$. $I_n^{2}[k]$ can be viewed as the worst case execution time
over the finite time horizon that our timing model applies. We assume that the constant value $I_n^{2}[k]$ approximate its true executing time well; $C_n^2[k]$ is the transmission duration of $\tau_n^{2}[k]$;  $\gamma_n[k]$ is the time instant when the transmission of $\tau_n^2[k]$ is completed; and $T_n[k]$ is the sampling interval between $\alpha_n[k]$ and $\alpha_n[k+1]$.  Then $\beta_n[k]$ represents the time when the sensor message finishes transmission, and $\gamma_n[k]$ represents the time when the control message finishes transmission. Note that the potential randomness and variations in $C_n^1[k]$ and $C_n^2[k]$, which are caused by the possible bit-stuffing, can be significantly reduced by effectively encoding the original payload\cite{Cena2012}. Even when the transmission time is NOT completely deterministic, the values of $C_n^1[k]$ and $C_n^2[k]$ will provide a good approximation of the actual transmission time. Here again we want to emphasize that the timing model applies to a finite time horizon only and is updated dynamically as part of the MPC scheme. So the small (unexpected) variations in the values of $C_n^1[k]$ and $C_n^1[k]$ will be tolerated by the control. There may also exist some general-purpose messages that are not related to the  control, but that share the CAN bus with the feedback control loops. These general-purpose messages can also be represented by  message chains. For example, one can let a message chain $\tau_{j}$  represent a general purpose message by choosing $I_{j}^2[k]=0$ and $C_{j}^{\rm 2}[k]=0$.
The following equations are satisfied by the parameters of a message chain when there is {\it no contention}:
\begin{align} \label{nocon}
\beta_n[k] &=\alpha_n(k)+I_n^1[k]+C_n^1[k] \cr
\gamma_n[k]&=\beta_n[k] +I_n^2[k]+C_n^2[k] \cr
\alpha_n[k+1]&= \alpha_n[k] + T_n[k]
\end{align}

The above equations will not hold when there is contention between messages.
Since only one message can be transmitted on the CAN bus at a time, $\tau_n^{1}[k]$ and $\tau_n^2[k]$ in $\tau_n[k]$ may not be transmitted immediately after they are generated. Instead, they have to compete with other messages for access to the CAN bus, under the CSMA/BA arbitration scheme. The priority of $\tau_n[k]$ can be represented by $P_n[k]$. Since each sub-message $\tau_n^1[k]$ and $\tau_n^2[k]$ in $\tau_n[k]$ may have its own priority, we have
\begin{equation}
P_n[k]=\left\{
\begin{array}{ll}
P_n^1[k] & {\rm when}\, \tau_n^{1}[k] {\rm \; is\,  transmitted} \vspace{1mm}\\
P_n^2[k] & {\rm when}\, \tau_n^{2}[k] {\rm \; is\,  transmitted}
\end{array}
\right.
\end{equation}
where $P_n^1[k]$ and $P_n^2[k]$ represent the priorities (identifier fields) of $\tau_n^1[k]$ and $\tau_n^2[k]$, respectively. We will see  in Section \ref{section:TimingModel} that equation (\ref{nocon}) will be replaced by a timing model which is able to answer the challenge.

\subsection{MPC Design}
MPC is an advanced control algorithm with increasing popularity in applications. It iteratively uses a  model of the feedback control loop to predict the future control strategy over a finite time horizon \cite{Rawlings2000}. However, only the first step of the predicted control strategy is implemented. At the next step, the process of predictions are repeated again, yielding a new control strategy. Such prediction horizon keeps shifting forward as time propagate.

For the $n$-th feedback control loop in Figure \ref{fig:network} where $n=1,2,...,N$, we assume that the plant is an independent, multiple input multiple output, and linear time-invariant system
\begin{align} \label{equation:plant}
\dot{x}_n(t)&=Ax_n(t)+Bu_n(t) \cr
y_n(t)&=Cx_n(t),
\end{align}
where $u_n(t)$ is the control command, $y_n(t)$ is the plant output, $x_n(t)$ is the plant state, and $A$, $B$ and $C$ are matrices of proper dimensions.

CAN based MPC relies on the controller nodes to compute the control effort $u_n(t)$ over a finite time window into the future. This finite time window is called the prediction horizon.
When a controller node is triggered by the end of the transmission of a sensor message in the same feedback control loop, the time when the sensor reading is obtained will be used as the start time of the prediction horizon. Denoting this start time by $t_0$, an  estimate  $\hat x_n(t_0)$ of the state is first obtained by a filtering algorithm. Let the finite prediction horizon be $[t_0, t_0+T_p]$, where $T_p$ is the length of the prediction horizon. The goal of the MPC is to find  control commands $u_n(t)$ that brings the predicted plant output $y_n(t)$ as close as possible to a reference trajectory $\lambda_n(t)$ for all $t\in [t_0, t_0+T_p]$.

A controller is triggered by the end of transmission of sensor messages, and an actuator node can only take actions when receiving a control message. Hence,  each model predictive  controller only needs to generate one control command for each sensor message received. The resulting control command is applied to the plant, and remains constant until the next sensor message triggers the controller again.  Time delay exists between the moment when the sensor takes measurements, and the moment when the actuator implements the control command. Therefore,  the control command $u_n(t)$ in   (\ref{equation:plant}) must be a piecewise   constant function
\begin{equation} \label{equation:delay}
u_n(t)\hspace{-0.7mm}=\hspace{-0.7mm}\mu_n[k],\, \; t\in \big[\alpha_n[k]\hspace{-0.7mm}+\hspace{-0.7mm}\delta_n[k], \alpha_n[k+1]\hspace{-0.7mm}+\hspace{-0.7mm}\delta_n[k\hspace{-0.7mm}+\hspace{-0.7mm}1]\big)\; ,
\end{equation}
where $\alpha_n[k]$ is the $k$-th sampling instant of the sensor as shown in Figure \ref{fig:NCtask}, and $\mu_n[k]$ is the optimal control command  that is generated by the model predictive controller in the sampling interval $[\alpha_n[k], \alpha_n[k+1]\,)$.  Also, $\delta_n[k]=\gamma_n[k]-\alpha_n[k]$ is the time delay between  the sampling time instant $\alpha_n[k]$, and the end time $\gamma_n[k]$ of the transmission of the control message, as shown in Figure \ref{fig:NCtask}.
Let $\mathbf{u}_n$ represents the piecewise constant control policy, defined by  $u_n(t)$, where $t \in [t_0, t_0+T_p]$.
If one can perform online prediction of $\delta_n[k]$ for all $k$ that fall within the
prediction horizon, then the piecewise constant control policy $\mathbf{u}_n$ is  a finite dimensional vector $[\mu_n[1], \mu_n[2],...,\mu_n[k],...]$ for
 all $k$ that fall within the finite prediction horizon $[t_0, t_0+T_p]$.

A cost function $J(x_n(t_0), \mathbf{u}_n)$ can be defined for the model predictive controller to optimize. One typical example of the cost function \cite{Liu2007}  is
\begin{equation} \label{equation:J}
\begin{array}{rcl}
&& J(x_n(t_0), \mathbf{u}_n) \\
&=&\int_{t_0}^{t_0+T_p} \big\{(\lambda_n(s)-y_n(s))^T Q_1 (\lambda_n(s)-
y_n(s))+\;u_n(s)^T Q_2 u_n(s)\big\}{\rm d} s \\
&& + x^T(t_0+T_p) Q_3 x (t_0+T_p),
\end{array}
\end{equation}
where $Q_1$, $Q_2$, and $Q_3$ are  positive semidefinite weighting matrices chosen by design. The first term in the integral penalizes the difference between the future plant output and the reference trajectory during the prediction horizon, and the second term   is the control penalty. The last term in the cost function is the terminal cost that ensures the system is stabilized by the controller. In (\ref{equation:J}), $y_n(t)$ must be predicted as a function of $x_n(t)$ and $u_n(t)$ for $t \in[t_0, t_0+T_p]$ through the process model in  (\ref{equation:plant})-(\ref{equation:delay}).

If the delays $\delta_n[k]$ for all tasks, indexed by $k$,  that falls within the interval $[t_0, t_0+T_p]$  can be predicted, then
 the model predictive control design problem can be formulated as a optimization problem that needs to compute at every $k$:
\begin{align} \label{equation:finaloptimization}
{\rm Given}\; x_n(t_0)=\hat x_n(t_0)\;{\rm and}\; \delta_n[k], \;\; {\rm solve}\; \; \min_{\mathbf{u}_n} J(x_n(t_0), \mathbf{u}_n)
\end{align}
subject to the following constraints:
\begin{equation}
\tag{\ref{equation:finaloptimization}.a}
\begin{array}{lc}
u_n(t)\in \mathcal{U},\;\; x_n(t)\in \mathcal{X},\;\;
\end{array}
\end{equation}
\begin{equation}
\tag{\ref{equation:finaloptimization}.b}
\dot{x}_n(t)=Ax_n(t)+Bu_n(t),\; \;
y_n(t)=Cx_n(t),\; \; {\rm and}
\end{equation}
\begin{equation}
\tag{\ref{equation:finaloptimization}.c}
u_n(t)\hspace{-0.7mm}=\hspace{-0.7mm}\mu_n[k],\, t\in \big[\alpha_n[k]\hspace{-0.7mm}+\hspace{-0.7mm}\delta_n[k], \alpha_n[k\hspace{-0.7mm}+\hspace{-0.7mm}1]\hspace{-0.7mm}+\hspace{-0.7mm}\delta_n[k\hspace{-0.7mm}+\hspace{-0.7mm}1]\big)\; ,
\end{equation}
where (\ref{equation:finaloptimization}.a) represents the constraints on the control command and the plant states. The sets $\mathcal{U}$ and $\mathcal{X}$ are assumed to be known.  Equations (\ref{equation:finaloptimization}.b) and (\ref{equation:finaloptimization}.c) represent the physical plant in the process model. The physical plant and the CAN timing model are coupled through the delay $\delta_n[k]$ in  (\ref{equation:finaloptimization}.c).
Note that in the cost function $J(x_n(t_0), \mathbf{u}_n)$, $y_n(t+\tau)$ for $\tau\in[0, T_p]$ must be predicted as a function of $x_n(t)$ and $u_n(t+\tau)$ for $\tau\in[0, T_p]$ through the process model in Equation (\ref{equation:finaloptimization}.b)  and  (\ref{equation:finaloptimization}.c).

If there was no contention on the CAN, the prediction of the time delays  $\delta_n[k]$ would be trivial. In fact, using equation (\ref{nocon}), the delay $\delta_n[k]=\gamma_n[k]-\alpha_n[k]=I_n^1[k]+C_n^1[k]+I_n^2[k]+C_n^2[k]$. In this special case the MPC design problem would be the classical problem which would be relatively easy to solve. We emphasize here that even in this special case, a continuous time MPC may be preferred than the discrete time one since the $\delta_n[k]$ may be time-varying.

\subsection{The Need of a Timing Model} \label{section:ProblemChallenge}
Real-time scheduling of messages under contention introduces time-varying delays $\delta_n[k]$ in Equation (\ref{equation:finaloptimization}.c). Since MPC design relies on the process model in Equation (\ref{equation:finaloptimization}.b) and (\ref{equation:finaloptimization}.c), the accurate prediction of $\delta_n[k]$ is important to MPC performance. Using the worst case delay would result in poor performance. Figure \ref{fig:Delayeffect} shows an example of MPC performance under either accurate or inaccurate prediction of $\delta_n[k]$. The inaccurate prediction of $\delta_n[k]$ is chosen as a constant delay from the worst-case analysis \cite{Tindell,Tindell1995, CanRevisit} and the accurate prediction of $\delta_n[k]$ is the actual time-varying delay of $\delta_n[k]$. The solid line represents the plant output $y_n(t)$ and the dashed line represents the reference trajectory $\gamma_n(t)$. As we can see, using an inaccurate $\delta_n[k]$ would lead to an unreliable process model, which severely degrades the performance of MPC.

\begin{figure}[htp]
\subfigure[Inaccurate prediction]
{\hspace{8mm}\includegraphics[trim = 15mm 0mm 9mm 7mm, clip, height=36mm, width=48mm]{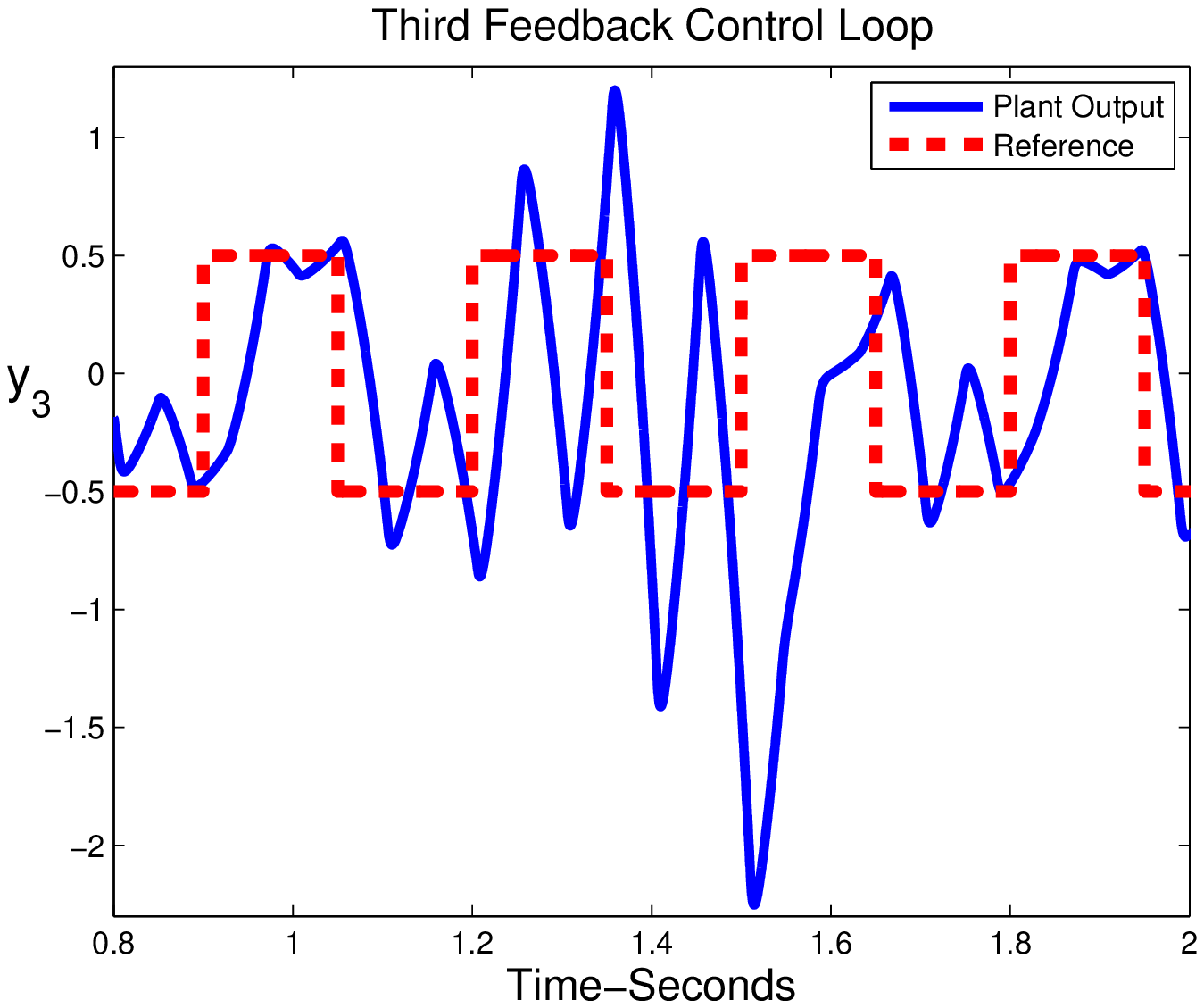}
  \label{fig:firstapproach}
 }
\subfigure[Accurate prediction]
{\hspace{9mm}\includegraphics[trim = 13.3mm 0mm 9mm 7mm, clip, height=36mm, width=48mm]{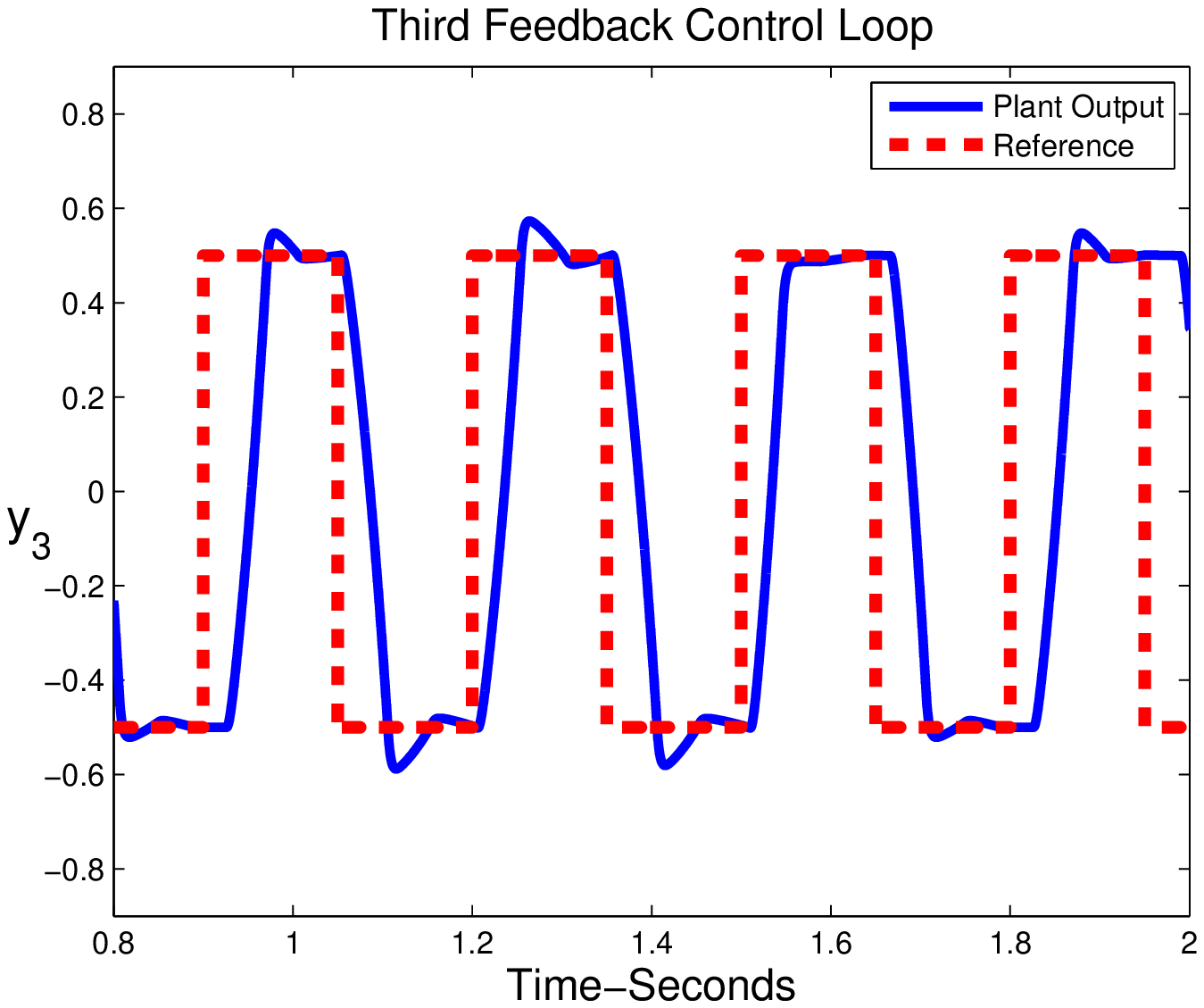}
  \label{fig:secondapproach}
 }
 \caption{MPC performance under two different predictions of $\delta_n[k]$} \label{fig:Delayeffect}
\end{figure}

MPC procedure treats the delay $\delta_n[k]$ as a timing constraint. The accurate prediction of $\delta_n[k]$ for messages under contention and priority based scheduling is difficult. To our knowledge, such model does not  exist  in the previous literature.
To answer this challenge, our contribution is to derive a timing model that is able to predict the timing constraints on the CAN-based control systems.

\section{The Timing Model } \label{section:TimingModel}

Our goal in this section is to derive a timing model for message chains under contentions that are resolved via the assigned priorities. This timing model generates predictions for  $\alpha_n[k]$, $\beta_n[k]$, and $\gamma_n[k]$ for all $n$ and $k$ for a finite length time window into the future, which will replace equation (\ref{nocon}) and then enable the MPC control design in (\ref{equation:finaloptimization}). Using the timing model, all transmission events, including the start and the end of all sensor and control messages can be
inferred.  From these timing information we will be able to estimate the delays $\delta_n[k]$ that are needed when computing the MPC.

Due to the time varying delays under contention, a continuous-time MPC design approach is a natural choice over discrete time MPC. To support the continuous time MPC,  we need to model the scheduled behaviors of message chains as a (piecewise) continuous function of $t$. Therefore, we redefine the message chain characteristics in continuous time domain as follows:
\begin{definition}
For any message chain $\tau_n$, an instance $\tau_n[k]$ is {\it active} at time $t$ if and only if it starts before $t$ and its next instance starts after $t$, i.e. $\alpha_n[k]<t<\alpha_n[k+1]$. At any time $t$, $\tau_n$ has only one active instance denoted as $\tau_n(t)$, i.e.
\begin{equation}
{\rm if }\; \alpha_n[k] \leq t <\alpha_n[k+1]\quad {\rm then\;}
\tau_n(t)\hspace{-0.1mm}=\hspace{-0.1mm}\tau_n[k]
\end{equation}
\end{definition}
\begin{definition}
At any time $t$, we define $\tau_n^1(t)$ and $\tau_n^2(t)$ as the first and second sub-messages in $\tau_n(t)$, i.e.
\begin{equation}
{\rm if }\; \alpha_n[k] \leq t <\alpha_n[k+1],
{\rm then\;} \tau_n^1(t)\hspace{-0.1mm}=\hspace{-0.1mm}\tau_n^1[k] \;\; {\rm and }\;\;\tau_n^2(t)\hspace{-0.1mm}=\hspace{-0.1mm}\tau_n^2[k] \\
\end{equation}
\end{definition}
Based on the above definitions, we can convert the message chain characteristics in Figure \ref{fig:NCtask} into a continuous time description for the active task instance $\tau_n(t)$.  $I_n^1(t)$ and $I_n^2(t)$ are the time needed for preparing $\tau_n^1(t)$ and $\tau_n^2(t)$.  $C_n^1(t)$ and $C_n^2(t)$ are the transmission duration of $\tau_n^1(t)$ and $\tau_n^2(t)$. $T_n(t)$ is the sampling interval of $\tau_n(t)$, and $P_n(t)$ is the priority of $\tau_n(t)$.
 These notations are summarized in Table 2.

\begin{table}[hbtp]
\label{notation}
\centering
\vspace{-2mm}
\caption{Characteristics of a message chain $\tau_n$}
\begin{tabular}{|c|c|} \hline
$\tau_n(t)$ & Active instance of $\tau_n$ at time $t$ \\
$\tau_n^1(t)$, $\tau_n^2(t)$ & 1$^{st}$ and 2$^{nd}$ sub-messages in $\tau_n(t)$ \\
$T_n(t)$ &  Sampling interval of $\tau_n(t)$\\
$P_n(t)$ & Priority of $\tau_n(t)$\\
$C_n^1(t)$, $C_n^2(t)$ & Transmission duration of $\tau_n^{1}(t)$, $\tau_n^{2}(t)$\\
$I_n^1(t)$, $I_n^2(t)$ & Time for preparing $\tau_n^{1}(t)$, $\tau_n^{2}(t)$\\
\hline
\end{tabular}
\end{table}

The parameters listed in Table 2 are not enough to describe the timing of message chains on the CAN due to  contention. The problem of scheduling message chains on a CAN  shares some similarity with the problem of task scheduling on a processor.  Authors of \cite{Zhang2013,Shi2013,Shi2012} introduced a dynamic timing model for the task scheduling problem on a processor. However, scheduling message chains on a CAN is a more complex problem. First, messages on the CAN are not preemptible while tasks considered in \cite{Zhang2013,Shi2013} are preemptible. Moreover, messages on the CAN are subject to causality constraints while tasks in \cite{Zhang2013,Shi2013,Shi2012} are independent. Such increased complexity requires significant extensions to the previous results. We show that the timing model will be a mixed set of continuous-time differential equations and logic equations that describe the evolutions of states that capture the timing.  This model can faithfully describe the timing of events.

\subsection{States of the Message Chains}
To model the preempted behaviors among multiple message chains, we introduce some extra parameters called the {\it states} for each message chain.
\begin{definition} \label{definition:D}
The {\it deadline} $d_n(t)$, for $n=1,2,\dotsb, N$, denotes how long after $t$ the next instance of the $n$th message chain  will start.
\end{definition}

\begin{definition} \label{definition:R}
The {\it residue} $r_n(t)$, for $n=1,2,...,N$, denotes the least remaining time required to finish processing and transmitting the active instance $\tau_n(t)$ {\em after} time $t$.
\end{definition}

\begin{definition} \label{definition:O}
The {\it delay} $o_n(t)$, for $n=1,2,...,N$, denotes the time between the starting time of $\tau_n(t)$ and the current time $t$ if the active instance $\tau_n(t)$ has not been fully processed. If the active instance
$\tau_n(t)$ has been fully processed at a time instant before the current time $t$, then the value of the delay  at time $t$ will be the length of the  time interval between the starting time of  $\tau_n(t)$  and the time instant when $\tau_n(t)$ has been fully processed.
\end{definition}

\begin{definition} \label{definition:ID}
The {\it index} $I\hspace{-0.5mm}D(t) \in\{1,\dotsb,N\}$ is the index of the message chain that is being transmitted on the CAN  at time $t$, where $I\hspace{-0.5mm}D(t)\neq 0$ implies that the active instance $\tau_{\scriptscriptstyle \hspace{-0.5mm}{I\hspace{-0.5mm}D\hspace{-0.2mm}(\hspace{-0.2mm}t\hspace{-0.2mm})}}$ is being transmitted and $I\hspace{-0.5mm}D(t)=0$ implies that no message chain is being transmitted.
\end{definition}

To help readers understand these concepts, let us consider the case of a message chain {\it without contention} as shown in Figure \ref{fig:NCtask}. Suppose the current time $t=\beta_n[k]$. Then
the deadline $d_n(t)=\alpha_n[k+1]-t=\alpha_n[k]+T_n[k]-\beta_n[k]$. The residue $r_n(t)=\gamma_n[k]-t=I_n^2[k]+C_n^2[k]$, and $o_n(t)=t-\alpha_n[k]=I_n^1[k]+C_n^1[k]$. These relationships will be much more complicated under contention.

We can assemble the states of all message chains  at time $t$ into a large row vector
 $Z(t)=[D(t), R(t), O(t), I\hspace{-0.5mm}D(t)] $ where $D(t)=[d_1(t),\dotsb, d_N(t)]$ ,  $R(t)=[r_1(t), ..., r_N(t)]$  and $O(t)=[o_1(t), ..., o_N(t)]$. Our timing model will determine the value of this row vector $Z(t)$ at any time $t$.

\subsection{Stages of a Message Chain} \label{section:Stages}
The residue $r_n(t)$ is a key state that indicates how much time is still needed before the active instance $\tau_n(t)$ will be completely processed. Its value always starts from
$I_n^1(t)+C_n^1(t)+I_n^2(t)+C_n^2(t)$ and decreases to $0$. During this process, the active instance $\tau_n(t)$  sequentially goes through seven different stages from the starting time to completion.
\begin{itemize}
\item  {\bf Stage 1}: the first sub-message $\tau_n^1(t)$ is being prepared. At this stage, $I\hspace{-0.5mm}D(t)\ne n$, the residue of $\tau_n(t)$ satisfies that $C_n^1(t)+I_n^2(t)+C_n^2(t)<r_n(t)\leq I_n^1(t)+C_n^1(t)+I_n^2(t)+C_n^2(t)$.\vspace{1.5mm}
\item {\bf Stage 2}: $\tau_n^1(t)$ is waiting for access to the CAN. At this stage, $I\hspace{-0.5mm}D(t)\ne n$,  the residue stays unchanged: $r_n(t)=C_n^1(t)+I_n^2(t)+C_n^2(t)$.  \vspace{1.5mm}
\item {\bf Stage 3}: $\tau_n^1(t)$ is being transmitted on the CAN. At this stage, $I\hspace{-0.5mm}D(t)=n$, the residue  satisfies that $I_n^2(t)+C_n^2(t)<r_n(t)< C_n^1(t)+I_n^2(t)+C_n^2(t)$. \vspace{1.5mm}
\item {\bf Stage 4}: the second sub-message $\tau_n^2(t)$ is being prepared. At this stage, $I\hspace{-0.5mm}D(t)\ne n$, the residue satisfies that $C_n^2(t)<r_n(t)\leq I_n^2(t)+C_n^2(t)$.\vspace{1.5mm}
\item {\bf Stage 5}: $\tau_n^2(t)$ is waiting for access to the CAN bus. At this stage, $I\hspace{-0.5mm}D(t)\ne n$,  the residue stays unchanged e.g. $r_n(t)=C_n^2(t)$.\vspace{1.5mm}
\item {\bf Stage 6}: $\tau_n^2(t)$ transmitting on the CAN bus. At this stage, $I\hspace{-0.5mm}D(t)= n$, the residue satisfies that $0<r_n(t)< C_n^2(t)$.\vspace{1.5mm}
\item {\bf Stage 7}: $\tau_n^2(t)$ is finished. At this stage, $I\hspace{-0.5mm}D(t)\ne n$, the residue  stays unchanged e.g. $r_n(t)=0$.
\end{itemize}
Whenever a new instance of $\tau_n$ arrives, it will go from Stage 7 back to Stage 1 and repeat the above process. Note that these stages are for one specific message chain. Multiple message chains may stay in different stages
at any given time.

Suppose the active instance of message chain $\tau_n(t)$ is marked by the index $k$.
The dynamic deadline $d_n(t)$ starts from the  initial value $T_n[k]$ and continuously decreases as time propagates, regardless of which stage the message chain is in. Hence we have that
\begin{equation}
\dot d_n(t) = -1,
\end{equation}
with initial value $d_n(\alpha_n[k])=T_n[k]$.
But after the value of $d_n(t)$ decreases to $0$, this indicates that a new instance of the message chain arrives. Then the message chain goes from Stage 7 back to Stage 1, and $d_n(t)$ will jump from $0$ to a new value $T_n[k+1]$.

The residue  $r_n(t)$ starts from the initial value $r_n(\alpha_n[k])= I_n^1[k]+C_n^1[k]+I_n^2[k]+C_n^2[k]$.
In Stages 2, 5, and 7, the residue satisfies $\dot r_n(t)=0$. In Stages 1, 3, 4, and 6, the residue decreases homogeneously e.g. $\dot r_n(t)=-1$.  The value of $r_n(t)$ will jump from $0$ to a new value $ I_n^1[k+1]+C_n^1[k+1]+I_n^2[k+1]+C_n^2[k+1]$ when a new instance of message arrives e.g. the message chain goes from Stage 7 back to Stage 1.

The delay $o_n(t)$ starts  from initial value $0$ at the starting time $\alpha_n[k]$ e.g. $o_n(\alpha_n[k])=0$. Whenever the value of the residue is not 0 e.g. $r_n(t)>0$, the delay
 increases homogeneously as $\dot o_n(t)=1$. In other words, the delay keeps increasing at Stages 1-6. The delay $o_n(t)$ stops increasing at Stage 7 since the active instance of the message chain
 has been fully processed. When a new message instance arrives e.g. the message chain goes from Stage 7 back to Stage 1, the delay $o_n(t)$ is reset to 0.

The index $I\hspace{-0.5mm}D(t)$ keeps constant until a change of access to the CAN happens. Since the CAN only transmits one message at a time, we have the following claim.
\begin{claim} \label{claim:Stages}
Consider a set of message chains $\{\tau_1, \dotsb, \tau_N\}$. At any time $t$, at most one message chain from $\{\tau_1, \dotsb, \tau_N\}$ can stay at Stage 3 or Stage 6, but multiple message chains can stay at other stages at the same time.
\end{claim}
The massage chain that is in Stage 3 or Stage 6 at the time $t$ will be the message indicated by the value of $I\hspace{-0.5mm}D(t)$. On a real CAN implementation, this value is known to all message chains due to the broadcasting mechanism used by the CAN. When one of the message chains is at Stage 3 or Stage 6, all other message chains will remain at Stages 1, 2, 4, 5, or 7. \CommentSZ{But} Since our goal is to derive a model for the CAN, we need to determine and predict the value of the state variable $I\hspace{-0.5mm}D(t)$ from the priorities $P_n(t)$. Let us suppose that a change of access to the CAN happens at a time instant $t_0$. From the values of the residue, we know which stage each message chain is at. Then the message that has access to the CAN will be the message with the highest priority among all messages that are either at stage 2 or stage 4 at time $t_0$. Therefore
\begin{equation}
I\hspace{-0.5mm}D(t_0)= \underset{\{i\vert \tau_i\, {\rm in\, Stages\, 2\, or\, 4,\, at}\, t_0\}}{\rm argmin}P_i(t_0).
\end{equation}
We enforce the convention that if the set $\{i\vert \tau_i\, {\rm in\, Stages\, 2\, or\, 4,\, at}\, t_0\}$ is empty, then $I\hspace{-0.5mm}D(t_0)=0$.
Note that this equation does not hold for all $t$ since the messages are non-preemptive. Therefore, to complete the timing model, we need to pinpoint the time instants when a change of access to the CAN happens.

As we see the evolution of  $Z(t)$ is relatively straightforward within each stage. What remains to do is to discover the length of each stage for each message chain. The length of Stages 1, 3, 4, and 6 are known due to the homogeneous decreasing of the residue $r_n(t)$. But the length of Stages 2, 5, and 7 can not be directly determined from the residue because it relies on knowing which message chain holds the access to the CAN.

\subsection{Significant Moments}
Let the current time be $t$, suppose the vector $Z(t)$ is completely known.  We need to predict the value of $Z(t+s)$ at a future time instant $t+s$.  We know that the values of $Z(t+s)$ will evolve continuously within each stage. However, since the message chain that has access to the CAN will change, and new instance of messages will arrive, the values of $Z(t+s)$ will not evolve continuously  in between different stages, but will rather have jumps.
The moments when these jumps happen are of  more significant value than other time instants.

\begin{definition}
At time $t$, we define the next {\it significant moment} as the time instant $t+S(t)$ where the state vector $Z(t)=\{D(t), R(t), O(t), I\hspace{-0.5mm}D(t)\}$ evolve continuously within the time interval $[\hspace{0.2mm}t, t+S(t))$, but sees a jump in one of the components of $Z(t)$ at time instant $t+S(t)$.
\end{definition}
The state vector $\{D(t), R(t), O(t), I\hspace{-0.5mm}D(t)\}$ evolves continuously most of time except in two situations: (1) a new message accesses the CAN and starts transmission, i.e. the message chain transits from Stage 2 to Stage 3 or from Stage 5 to Stage 6; and (2) a new instance of a message chain arrives, i.e. the message chain transits from Stage 7 to Stage 1. In the first situation, $I\hspace{-0.5mm}D(t)$ will have a jump; and in the second situation, components of the vector $\{D(t), R(t), O(t)\}$ will have a jump. At the current time $t$, the value of $S(t)$ is the time-interval between $t$ and the first time instant when a jump happens.

\subsubsection{A new message chain gaining access to CAN}
At the current time $t$, we want to know how long after $t$ a new message will gain access to the CAN.
Depends on whether the CAN is busy or idle at the current time $t$, we will have four different cases. To simplify the notation, we use $\tau_{I\hspace{-0.5mm}D}$ to denote $\tau_{I\hspace{-0.5mm}D(t)}$ in the following part of this paper, unless otherwise specified.

First, suppose that the CAN bus is busy at time $t$, i.e. $I\hspace{-0.5mm}D(t)\neq 0$, which implies that $\tau_{\scriptscriptstyle \hspace{-0.5mm}{I\hspace{-0.5mm}D}}$ is currently being transmitted on the CAN. As discussed in Section \ref{section:Stages}, we know that $\tau_{\scriptscriptstyle \hspace{-0.5mm}{I\hspace{-0.5mm}D}}$ at current time $t$ falls into either Stage 3 or Stage 6.
\vspace{1mm}

\noindent{\bf Case 1}: $\tau_{\scriptscriptstyle \hspace{-0.5mm}{I\hspace{-0.5mm}D}}$ at Stage 3 when the residue $r_{\scriptscriptstyle \hspace{-0.3mm} I\hspace{-0.5mm}D}(t)$ satisfies the following condition
\begin{equation}
I^2_{\scriptscriptstyle I\hspace{-0.5mm}D}(t)\hspace{-0.5mm}+\hspace{-0.5mm}C^2_{\scriptscriptstyle I\hspace{-0.5mm}D}(t)< r_{\scriptscriptstyle \hspace{-0.3mm} I\hspace{-0.5mm}D}(t)<C^1_{\scriptscriptstyle I\hspace{-0.5mm}D}(t)+I^2_{\scriptscriptstyle I\hspace{-0.5mm}D}(t)+C^2_{\scriptscriptstyle I\hspace{-0.5mm}D}(t)
\end{equation}
In this case, $\tau_{\scriptscriptstyle \hspace{-0.5mm}{I\hspace{-0.5mm}D}}$  will stay within Stage 3 before $\tau^1_{\scriptscriptstyle \hspace{-0.5mm}{I\hspace{-0.5mm}D}}$ finishing transmission. Hence  the next significant moment will happen no later than the moment when the transmission finishes. Therefore, $S(t)\le r_{\scriptscriptstyle \hspace{-0.5mm}{I\hspace{-0.5mm}D}}(t)\hspace{-0.9mm}-\hspace{-0.9mm}[I^2_{\scriptscriptstyle \hspace{-0.5mm}{I\hspace{-0.5mm}D}}(t)\hspace{-0.9mm}+\hspace{-0.9mm}C^2_{\scriptscriptstyle \hspace{-0.5mm}{I\hspace{-0.5mm}D}}(t)]$.

\vspace{2mm}
\noindent{\bf Case 2}: $\tau_{\scriptscriptstyle \hspace{-0.5mm}{I\hspace{-0.5mm}D}}$ at Stage 6, i.e. the residue $r_{\scriptscriptstyle \hspace{-0.3mm} I\hspace{-0.5mm}D}(t)$ satisfies the following condition
\begin{equation}
0< r_{\hspace{-0.3mm}\scriptscriptstyle I\hspace{-0.5mm}D}(t)< C^2_{\hspace{-0.3mm}\scriptscriptstyle I\hspace{-0.5mm}D}(t)
\end{equation}
In this case, $\tau_{\scriptscriptstyle \hspace{-0.5mm}{I\hspace{-0.5mm}D}}$  will stay within Stage 6 before $\tau^2_{\scriptscriptstyle \hspace{-0.5mm}{I\hspace{-0.5mm}D}}$ finishing transmission. No
other message will gain access to the CAN  before $\tau^2_{\scriptscriptstyle \hspace{-0.5mm}{I\hspace{-0.5mm}D}}$ finishing transmission. Then $S(t)\le r_{\scriptscriptstyle I\hspace{-0.5mm}D}(t)$.

\vspace{2mm}
\noindent Based on the above two cases, let us define $S_1(t)$ as the following
\begin{equation}\label{equation:remain_transmission}
S_1(t)= r_{\scriptscriptstyle \hspace{-0.5mm}{I\hspace{-0.5mm}D}}(t)\hspace{-0.8mm}-\hspace{-0.8mm}\left[I^2_{\scriptscriptstyle \hspace{-0.5mm}{I\hspace{-0.5mm}D}}(t)\hspace{-0.5mm}+\hspace{-0.5mm}C^2_{\scriptscriptstyle \hspace{-0.5mm}{I\hspace{-0.5mm}D}}(t)\right]\hspace{-0.5mm}{\rm sgn}({\rm max}\{0, r_{\scriptscriptstyle \hspace{-0.3mm} I\hspace{-0.5mm}D}(t)\hspace{-0.5mm}-\hspace{-0.5mm}C^2_{\scriptscriptstyle \hspace{-0.3mm} I\hspace{-0.5mm}D}(t)\}).
\end{equation}
Since the CAN can only transmit one message and the transmission is non-preemptive, it has to wait at least $S_1(t)$ amount of time before a new message can access to the CAN.
Then the next significant moment for  $\tau_{\scriptscriptstyle \hspace{-0.5mm}{I\hspace{-0.5mm}D}}(t)$ will be at $t+S(t)$ where $S(t) \le S_1(t)$.

Next, we suppose that the CAN is idle at time $t$, i.e. $I\hspace{-0.5mm}D(t)=0$, which implies no message is currently being transmitted on the CAN. In other words, all message chains are preparing sub-messages at current time $t$.  In this case, any message chain $\tau_n$ from $\{\tau_1, \dotsb,\tau_N\}$ falls into either Stage 1, 2, 4, or  5.  But if there is a message  at stage 2 or stage 5 and there is no other messages has access to the CAN, then this message will immediately gain access to the CAN right at the time $t$ and transits to Stage 3 or Stage 6 and $I\hspace{-0.5mm}D(t)\ne 0$. In these cases $S(t)=0$. Therefore, we only need to consider the cases where all message chains are either at Stage 1 or Stage 4.  Let us consider a message chain indexed by $n$.

\vspace{2mm}
\noindent{\bf Case 3}: $\tau_n$ is at Stage 1,  i.e. the residue $r_n(t)$ satisfies the following condition
\begin{equation}
C^1_n(t)+I^2_n(t)\hspace{-0.5mm}+\hspace{-0.5mm}C^2_n(t)< r_n(t)\leq I^1_n(t)+C^1_n(t)+I^2_n(t)+C^2_n(t)
\end{equation}
$\tau_n$ will stay within Stage 1 before $\tau^1_n$ finishing its preparation. The next significant moment will happen at least before  $\tau^1_n$ finishing its preparation.
Hence the value of $S(t)$ will be  no bigger than the remaining preparation time of $\tau^1_n$ e.g. $S(t)\le r_n(t)\hspace{-0.2mm}-\hspace{-0.2mm}C^1_n(t)\hspace{-0.2mm}-\hspace{-0.2mm}I^2_n(t)\hspace{-0.2mm}-\hspace{-0.2mm}C^2_n(t)$.

\vspace{2mm}
\noindent{\bf Case 4}: $\tau_n$ is at Stage 4, i.e. the residue $r_n(t)$ satisfies the following condition
\begin{equation}
C^2_n(t)< r_n(t)\leq I^2_n(t)+C^2_n(t)
\end{equation}
In this case, $\tau_n$ will stay within Stage 4 before $\tau^2_n$ finishing preparation. The next significant moment will happen at least before  $\tau^2_n$ finishing its preparation. Hence the value of
$S(t)$ will be no bigger than the remaining preparation time of $\tau^2_n$ e.g.  $S(t) \le r_n(t)-C_n^2(t)$.

\vspace{2mm}
\noindent  Based on the above two cases, we know that the next significant moment will happen at $t+S(t)$ where $S(t)$ should be at most equal to  the remaining preparation time for any message chain $\tau_n$
\begin{equation}
S(t)\le r_n(t)\hspace{-0.5mm}\hspace{-0.5mm}-\hspace{-0.5mm}C_n^2(t)\hspace{-0.5mm}-\hspace{-0.5mm}\left[C_n^1(t)\hspace{-0.5mm}+\hspace{-0.5mm}I_n^2(t)\right]{\rm sgn}({\rm max}\{0, r_n(t)\hspace{-0.6mm}-\hspace{-0.6mm}I_n^2(t)\hspace{-0.6mm}-\hspace{-0.6mm}C_n^2(t)\})
\end{equation}
This argument holds for all tasks in stages 1, 2, 4, or 5.
Define $S_2(t)$ as
\begin{equation} \label{equation:remain_preparation}
\begin{array}{c}
S_2(t)= \underset{1\leq n\leq N}{\rm min} \left\{r_n(t)\hspace{-0.5mm}\hspace{-0.5mm}-\hspace{-0.5mm}C_n^2(t)\hspace{-0.5mm}-\hspace{-0.5mm}\left[C_n^1(t)\hspace{-0.5mm}+\hspace{-0.5mm}I_n^2(t)\right]{\rm sgn}({\rm max}\{0, r_n(t)\hspace{-0.6mm}-\hspace{-0.6mm}I_n^2(t)\hspace{-0.6mm}-\hspace{-0.6mm}C_n^2(t)\})\right\}
\end{array}
\end{equation}
Therefore, the next significant moment will happen at $t+S(t)$ where $S(t)\le S_2(t)$.

At the significant moments $t+S(t)$ in the four cases above,  if the  either equation (\ref{equation:remain_transmission}) or equation (\ref{equation:remain_preparation}) holds e.g. $S(t)=S_1(t)$ for $I\hspace{-0.5mm}D(t)\neq 0$ or $S(t)=S_2(t)$ for $I\hspace{-0.5mm}D(t)=0$, then the values of $I\hspace{-0.5mm}D(t+S(t))$ will see a jump as
\begin{equation} \label{eq:jump2}
I\hspace{-0.5mm}D(t+S(t))= \underset{\{i\vert \tau_i\, {\rm in\, Stages\, 2\, or\, 4,\, at}\, t+S(t)\}}{\rm argmin}P_i(t+S(t))
\end{equation}
If the set $\{i\vert \tau_i\, {\rm at \, \, Stages\, 2\, or\, 4\, at}\, t+S(t)\}$ is empty, then $I\hspace{-0.5mm}D(t+S(t))=0$.
The values of $\{D(t), R(t), O(t)\}$ will remain unchanged.

\subsubsection{A new instance of message chain arrives}
The states of a message chain will  jump discretely whenever a new instance of a message arrives. For any message chain $\tau_n$, a new instance of $\tau_n$ will arrive at $t+d_n(t)$. Therefore, the earliest next instance of message chains in $\{\tau_1, \dotsb, \tau_N\}$ will not arrive until $t+\underset{1\leq n\leq N}{\rm min}\{d_n(t)\}$. Define $S_3(t)$ as
\begin{equation} \label{equation:time_interval_new}
S_3(t)= \underset{1\leq n\leq N}{\rm min}\{d_n(t)\}.
\end{equation}
Then $S(t) \le S_3(t)$.

Let $n^*$ be the index of the message chain that has the earliest instance that is arriving after $t$. If  $S(t)=d_{n^*}(t)$. Then
\begin{align} \label{eq:jump}
d_{n^*}(t+S(t))& =T_{n^*}(t+S(t)) \cr
r_{n^*}(t+S(t))& =I^1_{n^*}(t+S(t))+C^1_{n^*}(t+S(t)) +I^2_{n^*}(t+S(t))+C^2_{n^*}(t+S(t))  \cr
o_{n^*}(t+S(t)) &=0.
\end{align}
All the other components in $\{D(t+S(t)), R(t+S(t)), O(t+S(t))\}$ do not jump.
Since there is no change of access to the CAN,the state variable $I\hspace{-0.5mm}D(t+S(t))$ does not jump either.

\subsection{The Timing Model} \label{section:hybridmodel}
Let $S(t)={\min}\{S_1(t), S_2(t), S_3(t)\}$.
Our timing model integrates  both the continuous time evolution of the state vector $Z(t)$ within $[\hspace{0.2mm}t, t+S(t))$, and the  discrete jumps at $t+S(t)$. Hence the evolution of the state vector  within any large time interval $[t_a, t_b]$ can be obtained by concatenating the evolution within individual continuous time interval that belongs to $[t_a, t_b]$.

\begin{theorem} \label{claim:H}
At any time instant $t$, given initial values of the state vector $Z(t)=[D(t), R(t), O(t), I\hspace{-0.5mm}D(t)]$ and the parameters of the message chains $\{T_n(t+s'), I_n^1(t+s'), C_n^1(t+s'), I_n^2(t+s'), C_n^2(t+s'), P_n(t+s')\}_{n=1}^{N}$ for all $0\le s' \le s$, there exists a unique vector $[D(t+s), R(t+s), O(t+s), I\hspace{-0.5mm}D(t+s)]$.
\end{theorem}

\begin{proof}
Based on our previous discussion, we will just construct the unique solution $Z(t+s)$ at any $s>0$.
We first show that a unique trajectory is generated from the continuous evolution of the state vector $\{D(t), R(t), O(t), I\hspace{-0.5mm}D(t)\}$ from $t$ to any time $t+s$ where $t+s \in [\,t, t+S(t)\,)$.

For any message chain indexed $n$,  Since $d_n(t)$ will continuously decrease as time propagate, we have that
\begin{equation} \label{e1}
d_n(t+s) = d_n(t)-s
\end{equation}
Next, we consider the residue $r_n(t)$. If the message chain $n$ is at Stages 1, 3, 4, or 6\ then
\begin{equation} \label{e2}
r_n(t+s)=r_n(t)-s.
\end{equation}
If the message chain $n$ is at Stages 2, 5, or 7. Then
\begin{equation} \label{e3}
r_n(t+s)=r_n(t).
\end{equation}
Next, we consider the delay $o_n(t)$. If $\tau_n$ has been processed before $t$, i.e. $r_n(t)=0$, the delay $o_n(t)$ will not increase after $t$. On the other hand, if $\tau_n$ has not finished before $t$, i.e. $r_n(t)>0$, the delay $o_n(t)$ will continuously increase between $t$ and $t+s$. Thus, we have that
\begin{equation}
o_n(t+s) = o_n(t)+{\rm sgn}(r_n(t))\, s.
\end{equation}
Finally, we consider the index $I\hspace{-0.5mm}D(t)$. It will keep at constant between $t$ and $t+s$ since there is no significant moment, i.e.
\begin{equation}
I\hspace{-0.5mm}D(t+s)=I\hspace{-0.5mm}D(t).
\end{equation}
We see that all the values in the state vector $Z(t+s)$ are uniquely determined.

We now show that at a significant moment, the states jump to unique values. The possible values for $S(t)$ have been given in equations (\ref{equation:remain_transmission}), (\ref{equation:remain_preparation}) and (\ref{equation:time_interval_new}) as $S_1(t)$, $S_2(t)$ and $S_3(t)$. The possible jumps in the states are given by equations (\ref{eq:jump2}) and (\ref{eq:jump}). In all cases the states jump to unique values.
\qed
\end{proof}

Due to the theorem,  we can represent the hybrid timing model of the CAN based system as
\begin{align} \label{equation:hybridmodel}
Z(t+s)= \mathbb{H}\left(Z(t),  \{T_n,I_n^1, C_n^1, I_n^2, C_n^2,  P_n\}_{n=1}^N (t+s')\right)
\end{align}
 where the symbol $\mathbb{H}(\cdot)$ represents the timing model and  $\{T_n,I_n^1, C_n^1, I_n^2, C_n^2,  P_n\}_{n=1}^N (t+s')$ represents the parameters of all
 tasks at any time $t+s'$ for all $0\le s' \le s$.

One immediate benefit of this timing model is a necessary and sufficient condition for schedulability of all messages in a finite time window.
\begin{definition}
A message chain $\tau_n$ is {\it instantaneously schedulable} on the CAN at time $t$ if $r_n(t)\le d_n(t)$.
\end{definition}
If $\tau_n$ is instantaneously schedulable for all time $t$, then all the deadlines of $\tau_n$ are met, then message $\tau_n$ is schedulable in the usual definition.
On the other hand, if the message chain $\tau_n$ is schedulable, then all the deadlines of $\tau_n$ are met, which implies that the message chain is instantaneously schedulable
for all $t$.

\begin{corollary}
A message chain $\tau_n$ is  instantaneously schedulable on the CAN at time $t$ if   $r_n((t+S(t))^-)\le d_n((t+S(t))^-)$.
\end{corollary}
\begin{proof}
Using the dynamic timing model which contains equations (\ref{e1}), (\ref{e2}) and ({\ref{e3}}), we must have
\begin{align}
	d_n((t+S(t))^-) - r_n((t+S(t))^-)&= d_n(t)-S(t) -  r_n((t+S(t))^-) \cr
	& \le  d_n(t)-S(t) -  (r_n(t)-S(t)) \cr
	&= d_n(t) - r_n(t).
\end{align}
Hence if $r_n((t+S(t))^-)\le d_n((t+S(t))^-)$, then $r_n(t)\le d_n(t)$.
\end{proof}

\begin{corollary}
A message chain is schedulable if and only if it is instantaneously schedulable at the significant moments.
\end{corollary}

\begin{proof}
Consider the time instants right before
the significant moments $t+S(t)$. If the message is instantaneously schedulable at these moments, then the message chain is instantaneously schedulable at any time $t$. The entire message change is schedulable. If a message chain is not instantaneously
schedulable at the significant moments, then the message chain is not schedulable. \qed
\end{proof}

\section{State Observer} \label{section:StateObserver}
At each embedded controller node,  the hybrid timing model will be used to predict delays and timing constraints for MPC. The prediction requires the knowledge of the state vector $Z(t)=[D(t), R(t), O(t), I\hspace{-0.5mm}D(t)]$. Since the CAN uses a broadcast scheme, each embedded controller node will know which message is currently being transmitted on the CAN, i.e. the value of $I\hspace{-0.5mm}D(t)$ can be determined. However, the values of the states $[D(t), R(t), O(t)]$ may not be measured directly. In this section, we will discuss how to estimate the state vector $[D(t), R(t), O(t)]$ based on events that can be observed on the CAN.

\subsection{Estimation of $[D(t), R(t), O(t)]$}
As discussed in \cite{DiNatale2012}, CAN chips can generate an interrupt whenever a message is received by a node. These interrupts can be pre-handled by a dedicated MCU that usually shipped together with CAN chip. Therefore, we can easily design an interrupt handler on the host processor of a CAN node to observe the receiving times of $\tau_n^1[k]$ and $\tau_n^2[k]$, which corresponds to $\beta_n[k]$ and $\gamma_n[k]$ as shown in Figure \ref{fig:NCtask}.  Note that the CAN  utilizes a broadcast scheme for message transmission. The MPC controller node in each feedback loop can not only receive messages within its own control loop, but also messages from other feedback control loops. Therefore, each MPC controller node has complete information of $\{\beta_n[k], \gamma_n[k]\}$ for all message chains $\{\tau_1,\dotsb,\tau_N\}$ on the CAN. But there is no direct way to measure $\alpha_n[k]$.

Based on the above observations, we propose an algorithm to estimate the value of $\alpha_n[k]$ as follows
\begin{equation} \label{equation:alphaestimation}
\begin{array}{c}
\hat{\alpha}_n[k]=
{\rm min}\{ \hat{\alpha}_n[k\hspace{-0.9mm}-\hspace{-0.9mm}1]\hspace{-0.9mm} + \hspace{-0.9mm}T_n[k\hspace{-0.9mm}-\hspace{-0.9mm}1],\,\beta_n[k]\hspace{-0.9mm}-\hspace{-0.9mm}C_n^1[k]\hspace{-0.9mm}-\hspace{-0.9mm}I_n^1[k]\}
\end{array}
\end{equation}
where $\hat{\alpha}_n[k-1]$ is the estimate from the previous observations of $\beta_n[k-1]$ and $\gamma_n[k-1]$. Each controller node can estimate $\hat{\alpha}_n[k]$ for all message chains. The computation of $\hat{\alpha}_n[k]$ for $1\leq n\leq N$ at each node is linear with respect to the number of control loops.

At the current time $t$, given $\{\hat{\alpha}_n[k], \beta_n[k], \gamma_n[k]\}$, each embedded controller node can estimate the state vector $\{\hat{d}_n(t), \hat{o}_n(t), \hat{r}_n(t)\}$.
The deadline $d_n(t)$ is estimated as
\begin{equation} \label{e5}
\hat{d}_n(t)=\hat{\alpha}_n[k]+T_n[k]-t,
\end{equation}
where $\hat{\alpha}_n[k]+T_n[k]$ is the time instant when $\tau_n[k+1]$ starts.  The delay $o_n(t)$ is estimated as
\begin{equation} \label{equation:o_estimation}
\hat{o}_n(t)=
\left\{
\begin{array}{cl}
t-\hat{\alpha}_n[k] & \qquad  \; \mbox{ if }\tau_n^2[k] \;{\rm \;NOT\;received}\vspace{0.5mm}\\
\gamma_n[k]-\hat{\alpha}_n[k] & \qquad  \; \mbox{ if }\tau_n^2[k] \;{\rm \;received}\vspace{0.5mm}\\
\end{array}
\right.,
\end{equation}
 The delay will not increase if $\tau_n^2[k]$ has finished transmission before $t$. The residue $r_n(t)$ can be estimated as
\begin{equation}
\label{equation:r_estimation}
\begin{array}{c}
\hat{r}_n(t)=\\
\left\{
\begin{array}{ll}
\{I_n^1\hspace{-1mm}+\hspace{-1mm}C_n^1\hspace{-1mm}+\hspace{-1mm}I_n^2\hspace{-1mm}+\hspace{-1mm}C_n^2\}[k]\hspace{-1mm}-\hspace{-1mm}{\rm min}\{\,t\hspace{-1mm}-\hspace{-1mm}\hat{\alpha}_n[k], I_n^1[k]\,\},  & \tau_n^1[k]\;{\rm and}\; \tau_n^2[k] \;{\rm NOT\;received}\vspace{1.5mm}\\
I_n^2[k]\hspace{-1mm}+\hspace{-1mm}C_n^2[k]\hspace{-1mm}-\hspace{-1mm}{\rm min}\{t\hspace{-1mm}-\hspace{-1mm}\beta_n[k], I_n^2[k]\},  & \tau_n^1[k]\;{\rm received}, \tau_n^2[k]\;{\rm NOT\;received} \vspace{1mm}\\
0,  &  \tau_n^1[k]\;{\rm and}\; \tau_n^2[k] \;{\rm received}
\end{array}
\right.
\end{array}
\end{equation}
where $\{I_n^1\hspace{-1mm}+\hspace{-1mm}C_n^1\hspace{-1mm}+\hspace{-1mm}I_n^2\hspace{-1mm}+\hspace{-1mm}C_n^2\}[k]$ is the shorthand notation for $I_n^1[k]\hspace{-1mm}+\hspace{-1mm}C_n^1[k]\hspace{-1mm}+\hspace{-1mm}I_n^2[k]\hspace{-1mm}+\hspace{-1mm}C_n^2[k]$.

Whenever a message is received by the controller node, an interrupt function can be triggered to estimate $[\hat d_n(t), \hat r_n(t), \hat o_n(t)]$ for $n=1,2,...,N$ at  the moment of reception. Then the state vector $[\hat{D}(t), \hat{R}(t), \hat{O}(t)]$ will be constructed. The timing model $\mathbb{H}$ can then be used to predict the state vectors in future times starting from $t$.

\subsection{Convergence of Estimation}
We show that the estimation $[\hat{D}(t), \hat{R}(t), \hat{O}(t)]$ will have bounded error. The error will not increase as time $t$ propagates.

As we discussed in Equation (\ref{e5}), (\ref{equation:o_estimation}), and (\ref{equation:r_estimation}), the estimates $[\hat{D}(t), \hat{R}(t), \hat{O}(t)]$ are derived from $\{\hat{\alpha}_n[k], \beta_n[k], \gamma_n[k]\}_{n=1}^{N}$. Since $\{\beta_n[k], \gamma_n[k]\}_{n=1}^{N}$ can be directly observed from the CAN, the accuracy of estimating $[\hat{D}(t), \hat{R}(t), \hat{O}(t)]$ is actually determined by the accuracy of estimating $\hat{\alpha}_n[k]$.
Define the estimation error between $\hat{\alpha}_n[k]$ and $\alpha_n[k]$ as
\begin{equation} \label{equation:epsilon}
\epsilon_n[k]=\hat{\alpha}_n[k]-\alpha_n[k]  \;{\rm for \; any\;} k\ge0
\end{equation}

\begin{claim} \label{claim:epsilon_converge}
The estimation error $\epsilon_n[k]$ is non-negative and non-increasing as $k$ grows, i.e.
\begin{equation}
\epsilon_n[0]\ge\epsilon_n[1]\ge \dotsb \ge \epsilon_n[k]\ge\epsilon_n[k+1]\ge\dotsb\ge0 .
\end{equation}
\end{claim}

\begin{proof}
First, we prove that the estimation error is non-negative, i.e. $\epsilon_n[k]\ge 0$ for any $k\ge 0$. When multiple message chains $\{\tau_1, \dotsb, \tau_N\}$ are  transmitted on the CAN, each message may not be transmitted immediately after it is ready. Instead, it has to compete with other messages for access to the CAN. Thus, we have that
\begin{equation} \label{equation:alpha1}
\alpha_n[k]+I_n^1[k]\leq \beta_n[k]-C_n^1[k] \;{\rm for \; any\;} k\ge0
\end{equation}
where the left hand side represents the time when a message $\tau_n^1[k]$ is ready for transmission, i.e. $\tau_n$ at Stage 2, and the right hand side represents the time when $\tau_n^1[k]$ actually starts to transmit on the CAN bus, i.e. $\tau_n$ at the beginning of Stage 3. According to Equation (\ref{equation:alphaestimation}) and (\ref{equation:alpha1}), we know that
\begin{equation}
\hat{\alpha}_n[0]=\beta_n[0]-C_n^1[0]-I_n^1[0]\ge \alpha_n[0]
\end{equation}
which implies $\epsilon_n[0]\ge0$. Moreover, we have that
\begin{equation} \label{equation:alpha2}
\hat{\alpha}_n[0]+T_n[0]\ge \alpha_n[0]+T_n[0]=\alpha_n[1].
\end{equation}
According to Equation (\ref{equation:alpha1}), we have that
\begin{equation}  \label{equation:alpha3}
\beta_n[1]-C_n^1[1]-I_n^1[1] \ge \alpha_n[1]
\end{equation}
Therefore, based on Equation (\ref{equation:alphaestimation}), (\ref{equation:alpha2}), and (\ref{equation:alpha3}), we have that
\begin{equation} \label{equation:c3}
\hat{\alpha}_n[1]\hspace{-0.5mm}=\hspace{-0.5mm}
{\rm min}\{\hat{\alpha}_n[0]\hspace{-0.9mm} + \hspace{-0.9mm}T_n[0],\,\beta_n[1]\hspace{-0.9mm}-\hspace{-0.9mm}C_n^1[1]\hspace{-0.9mm}-\hspace{-0.9mm}I_n^1[1]\}\ge \alpha_n[1]
\end{equation}
which implies that $\epsilon_n[1]\ge0$. By induction, we have  shown that $\epsilon_n[k]\ge 0$ for any $k\ge 0$.

Next, we show that the estimation error $\epsilon_n[k]$ is non-increasing as $k$ grows, i.e. $\epsilon_n[k]\ge \epsilon[k+1]$. According to Equation (\ref{equation:alphaestimation}), we have that
\begin{equation}
 \hat{\alpha}_n[k+1]\leq \hat{\alpha}_n[k]+T_n[k]
\end{equation}
which implies that
\begin{equation}
\hat{\alpha}_n[k+1]-\hat{\alpha}_n[k]\leq T_n[k]=\alpha_n[k+1]-\alpha_n[k]
\end{equation}
Hence, we have that
\begin{equation}
\hat{\alpha}_n[k]-\alpha_n[k]\ge\hat{\alpha}_n[k+1]-\alpha_n[k+1]
\end{equation}
Therefore, $\epsilon[k]\ge \epsilon[k+1]$ for any $k \ge 0$ is proved. \qed
\end{proof}

The claim implies that the estimation error for the state vector are all bounded and the error will never increase.  In fact, we have observed in our simulations that this error often decreases to zero. But there are cases where the error stays as a constant value.

Using the estimated states, we can also test for instantaneous schedulability by checking the condition $\hat r_n(t) \le \hat d_n (t)$ at the significant moments.  The following theorem holds.
\begin{theorem}
If a task is instantaneously schedulable e.g. $r_n(t) \le d_n(t)$, then the estimated states satisfies $\hat r_n(t) \le \hat d_n (t)$.
\end{theorem}
\begin{proof}
	According to equation (\ref{e5}), we have
	\begin{align}
		\hat{d}_n(t)&=\hat{\alpha}_n[k]+T_n[k]-t  \cr
			&= \alpha_n[k] + \epsilon_n[k] +T_n[k]-t  \cr
			&=  	 d_n(t) + \epsilon_n[k] \cr
			&\ge r_n(t) + 	 \epsilon_n[k].
	\end{align}
	According to (\ref{equation:r_estimation}), we have
\begin{equation}
\begin{array}{c}
\hat r_n(t) - r_n(t)=\\
\left\{
\begin{array}{ll}
\hspace{-1mm}{\rm min}\{\,t\hspace{-1mm}-\hspace{-1mm}{\alpha}_n[k], I_n^1[k]\,\}\hspace{-1mm}-\hspace{-1mm}{\rm min}\{\,t\hspace{-1mm}-\hspace{-1mm}\hat{\alpha}_n[k], I_n^1[k]\,\} & \tau_n^1[k]\;{\rm and}\; \tau_n^2[k] \;{\rm NOT\;received}\vspace{1.5mm} \\
0 & {\rm otherwise}
\end{array}
\right.
\end{array}
\end{equation}
which implies that $\hat r_n(t) - r_n(t)=\hat \alpha_n(t) - \alpha_n(t)=\epsilon_n[k]$. Therefore
	$\hat r_n(t) \le \hat d_n (t)$. \qed
\end{proof}
The above theorem implies that if the message chains  are schedulable, then the estimated states will never fail the schedulability test. On the other hand, suppose we detect that a message chain is not schedulable using the estimated states, then the task set must not be schedulable.

\section{MPC Design} \label{section:MPC}
In this section, the MPC design problem proposed in Section \ref{section:Problem} will be solved. We assume that all the message chains are schedulable. Since each control loop is independent, the MPC design for any of the loops can be solved in the same way.

Let us consider the MPC design for the $n$th feedback loop corresponding to the message chain $\tau_n$. As discussed in Section \ref{section:ProblemChallenge},  the value  of $\delta_n[k]$ within the
prediction horizon is needed for MPC design. The message chain $\tau_n$ has $K$ instances that falls within the prediction horizon. Let the indices  of these instances starts from $k$ and ends at $k+K-1$ where $K\ge 1$ is an integer.
 Then we need to determine $\delta_n[k+j-1]$ for $j=1,2,...,K$.

\begin{theorem}
Suppose all messages are schedulable.   Consider the time instants
\begin{equation}
t_j = \CommentSZ{\hat}{\alpha}_n[k] + \sum_{l=1}^{j} T_n[k+l-1]
\end{equation}
for $j=1,2,...,K$. Then the delay $\delta_n[k+j-1]$ can be obtained from the states as
\beq  \label{equation:timingmodel2}
	\delta_n[k+j-1] = o_n(t_j^-).
\eeq
\end{theorem}
\begin{proof}
By definition of the state variable $o_n(t)$, it represents the time delay between the starting time of the active instance of a message chain and the time $t$. If we let $t=t_j^-$, then $o_n(t_j^-)$ is the delay between
the starting time of the active instance $\tau(t_j)$ and $t_j$. Since all message chains are schedulable, the active instance $\tau(t_j)$  would have been processed before $t_j$. Then the  delay $o_n(t_j^-)$ is the delay between the starting time and the finishing time of the active instance e.g. $\delta_n[k+j-1] = o_n(t_j^-)$. \qed
\end{proof}

Let the current time be $t$, suppose we have  estimated the state vector $\hat{Z}(t)$ by the state observer introduced in Section \ref{section:StateObserver}. Then, we will be able to predict the future trajectory of $\hat{Z}(t+s)$ for all $s\in [0, T_p]$ where $T_p$ is the length of the prediction horizon:
\begin{equation} \label{equation:timingmodel1}
\begin{array}{c}
[\hat{D}(t+s), \hat{R}(t+s), \hat{O}(t+s), I\hspace{-0.5mm}D(t+s)]=\vspace{1.5mm}\\
\mathbb{H}(\hspace{0.3mm}[\hat{D}(t), \hat{R}(t), \hat{O}(t), I\hspace{-0.5mm}D(t)], \{T_n, I_n^1, C_n^1,I_n^2,C_n^2,P_n\}_{n=1}^{N}(t+s'))
\end{array}
\end{equation}
where $0\le s' \le s$.
 Using the hybrid timing model $\mathbb{H}$, let $t+s=t_j$ for $j=1,2,...,K-1$. we can perform online prediction of the delay as $\delta_n[k+j-1]=\hat o_n(t_j)$ according to  equation (\ref{equation:timingmodel2}).
 Due to the fact that $\hat \alpha_n[k] \ge \alpha_n[k]$, the delay based on the estimate $\hat o_n$ may be smaller than the actual delay.

With the delay $\delta_n[k+j-1]$ determined for $j=1,2,...,K$ all determined, the MPC design problem (\ref{equation:finaloptimization}) subject to the constraints
(\ref{equation:finaloptimization}.a)-(\ref{equation:finaloptimization}.c) is now well formulated. The solution of the continuous time MPC problem can be obtained using well-known optimization techniques as in \cite{Wang2009}. The resulting piecewise linear control effort is then applied to the plant until the next time the controller is triggered.  The timing model will be engaged again to predict the delays, and then a new piecewise control law will be computed by solving the MPC design problem. This process will be iterated. The prediction of the delay requires little computing time for the following reasons: (1) the timing model is very simple with linear complexity; (2) the calculation is only performed at the significant moment because the transition between any two consecutive moments is continuous and follows the equations in the timing model. Hence, the timing model is compatible with the MPC design approach.

\section{Numeric Simulation} \label{section:simulation}
In this section, we use numeric simulations to demonstrate  the MPC design using the hybrid timing model of the CAN. We show that the timing model is preferred even when there exist other simulation tools to generate the timing sequences for the message chains.

The simulation environment for the CAN-based control system is established according to Figure \ref{fig:network}. To compare with our approach,
the CAN in Figure \ref{fig:network} is simulated using Truetime (Version 2.0) \cite{Cervin2003}. Truetime is a Matlab/Simulink-based simulator for real-time control system, which provides a network block that supports the protocol of the CAN.  The Truetime simulation results are used as the ground truth for the timing of message chains.

Our simulation contains  three feedback control loops sharing the CAN. The plant in each feedback loop is an inverted pendulum model represented as follows
\begin{align} \label{equation:pendulum}
\dot{x}_n(t)&=\left[
\begin{array}{cc}
0 & 1\\
a_n & b_n
\end{array}
\right]
x_n(t)+\left[
\begin{array}{c}
0\\
c_n
\end{array}
\right]u_n(t) \cr
y_n(t)&=\left[\begin{array}{cc}
1 & 0
\end{array}\right]x_n(t) \cr
\end{align}
The inverted pendulums in the three feedback control loops have different coefficients as $[a_1, a_2, a_3]=[98, 65, 44]$, $[b_1, b_2, b_3]=[120, 52, 30]$ and $[c_1, c_2, c_3]=[20, 13, 10]$. The sensor nodes sample the state of the plants at the time interval of $20\,$ms, $30\,$ms, and $40\,$ms. Each sensor node needs $1\,$ms to process the sampling information and generate a sensor message. The MPC controller node in each feedback control loop computes an optimal control signal $u_n(t)$ that makes the plant output $y_n(t)$ track a given reference trajectory $\gamma_n(t)$ as close as possible, under the constraint that $-4 \leq u_n(t) \leq 4$. The computation time of an MPC is $2\,$ms. The actuator node takes action as soon as the control message is received from the CAN bus. We assume that sensor and control messages have the transmission duration of $3\,$ms and they are assigned unique identifier fields such that the priorities of the message chains satisfy $P_1^{1}[k]<P_1^{2}[k]<P_2^{1}[k]<P_2^{2}[k]<P_3^{1}[k]<P_3^{2}[k]$ which implies that the first feedback loop has the highest priority and the third loop has the lowest priority.  Hence the three message chains  transmitted on the CAN  have the following characteristics
\begin{align} \label{equation:nominalcharacteristics}
\left[\hspace{0mm}T_1(t), I_1^1(t), \hspace{0mm}C_1^1(t), \hspace{0mm}I_1^2(t),\hspace{0mm}C_1^2(t)  \hspace{0mm}\right]\hspace{-0.6mm}=\hspace{-0.6mm}[20, 1, 3, 2, 3]\,{\rm ms} \cr
\left[\hspace{0mm}T_2(t), I_2^1(t), \hspace{0mm}C_2^1(t), \hspace{0mm}I_2^2(t),\hspace{0mm}C_2^2(t)  \hspace{0mm}\right]\hspace{-0.6mm}=\hspace{-0.6mm}[30, 1, 3, 2, 3]\,{\rm ms} \cr
\left[\hspace{0mm}T_3(t), I_3^1(t), \hspace{0mm}C_3^1(t), \hspace{0mm}I_3^2(t),\hspace{0mm}C_3^2(t)  \hspace{0mm}\right]\hspace{-0.6mm}=\hspace{-0.6mm}[40, 1, 3, 2, 3]\,{\rm ms}
\end{align}

\subsection{Verification of Hybrid Timing Model}
We  first  verify the correctness of our proposed timing model by  comparing the delays predicted through the hybrid timing model with the delay observed from the simulation results generated from Truetime.
Suppose the message chains in Equation (\ref{equation:nominalcharacteristics}) are being transmitted on the CAN. Figure \ref{fig:Truetime} shows the timing of message chains generated by the Truetime simulation. Table \ref{fig:hybriddelay} shows the delays $\delta_n[k]$ predicted through the hybrid timing model in Equation (\ref{equation:timingmodel1}) and (\ref{equation:timingmodel2}). In Figure \ref{fig:Truetime}, the value ``0.5". indicates that the message is ready for transmission but blocked by other messages on the CAN bus, the value ``1" indicates that the message is being transmitted on the CAN bus, and the value ``0" indicates that the message finishes transmission.

\begin{table}
\centering
\begin{tabular}{ |c|c|c|c|c|c|}
\hline
$\delta_n[k]$ & k=1 & k=2 & k=3 & k=4\\ \hline
n=1 & 10\,ms & 9\,ms & 10\,ms &10\,ms  \\ \hline
n=2 & 13\,ms & 9\,ms & 13\,ms & 11\,ms \\ \hline
n=3 & 21\,ms & 13\,ms & 13\,ms & 21\,ms  \\
\hline
\end{tabular}
\caption{Delays predicted through the hybrid timing model}
\label{fig:hybriddelay}
\end{table}

\begin{figure}[tb]
\hspace{-1mm}
\centering
\includegraphics[width=0.49\textwidth]{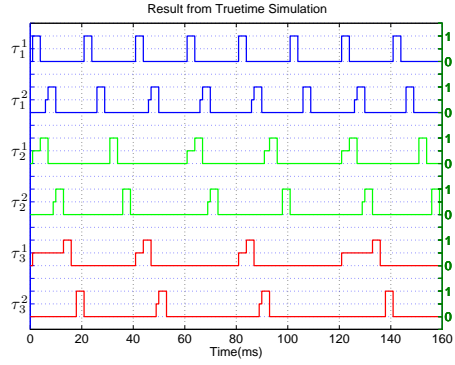}
\caption{Timing of message chains produced by Truetime simulation} \label{fig:Truetime}
\end{figure}

For illustration, we examine the delay $\delta_3[k]$ in the third feedback control loop. The delays in other feedback control loops can be studied using the exactly same procedure. We know that $\delta_3[k]$ is a time interval between the moment when the sensor take measurements and the moment when the actuator take actions.  The sensor in the third feedback control loops take measurements at $0\,$ms, $40\,$ms, $80\,$ms, and $120\,$ms.  By closely examining Figure \ref{fig:Truetime}, we observe that the control message $\tau_3^2$ in the third feedback control loop  finishes transmission at $21\,$ms, $53\,$ms, $92\,$ms, and $141\,$ms. Therefore, the observation of Figure \ref{fig:Truetime} shows that the value of $\delta_3[k]$ is $21\,$ms, $13\,$ms, $12\,$ms, and $21\,$ms, for $1\leq k\leq 4$. This observation exactly matches the value of $\delta_3[k]$ listed in Table \ref{fig:hybriddelay}. Similarly, we can see that the values of $\delta_1[k]$ and $\delta_2[k]$ observed from Figure \ref{fig:Truetime} also match that listed in Table \ref{fig:hybriddelay}. Therefore, we can claim that the hybrid timing model can accurately describe the timing of message chains on the CAN.

\subsection{Analysis of Computational Cost}
Even though Truetime and other event-based simulation tools  are able to generate the timing sequences of the message chains, running such simulation takes significant amount of computation resources. Hence these simulations may be too slow for realtime embedded applications. Our timing model is discontinuous at limit number of time points, but continuous the rest of time. So, running our model  only requires significant computation at a small fraction of discrete time points, and the system transition between any two consecutive discrete time points can be directly derived using mathematical equations. This has caused a significant reduction of computing load when compared to typical simulation based methods. To verify
this computational advantage, we evaluate the computational time of generating scheduled behavior in Figure \ref{fig:Truetime} using both the hybrid timing model and Truetime. The experiment is performed on a MacBook computer with Processor 2.26 GHz Intel Core 2 Duo, and Memory 4GB 1067MHz DDR3. Since Truetime is written in C++ Mex, we also implement the analytical timing model in the same way as Truetime. Matlab version 2010Rb and the Trutime Version 2.0 are used for the comparison. For each simulation window length that falls within [0, 100]s, we run both methods 50 times and then calculate the averaged computation time for each method. Fig \ref{fig:ComputationalCost} shows the comparison. The horizontal axis denotes the window length used for all simulated scheduled behaviors, and the vertical axis denotes the time spent to compute the simulation. In both figures, the computational time linearly increases with window length. More importantly, we can see that the hybrid timing model is approximately 4000 times faster than Truetime, which is a significant improvement for embedded system applications.

\begin{figure}[htp]
\subfigure[Hybrid Timing Model]
{\hspace{-1mm}\includegraphics[height=44mm, width=60mm]{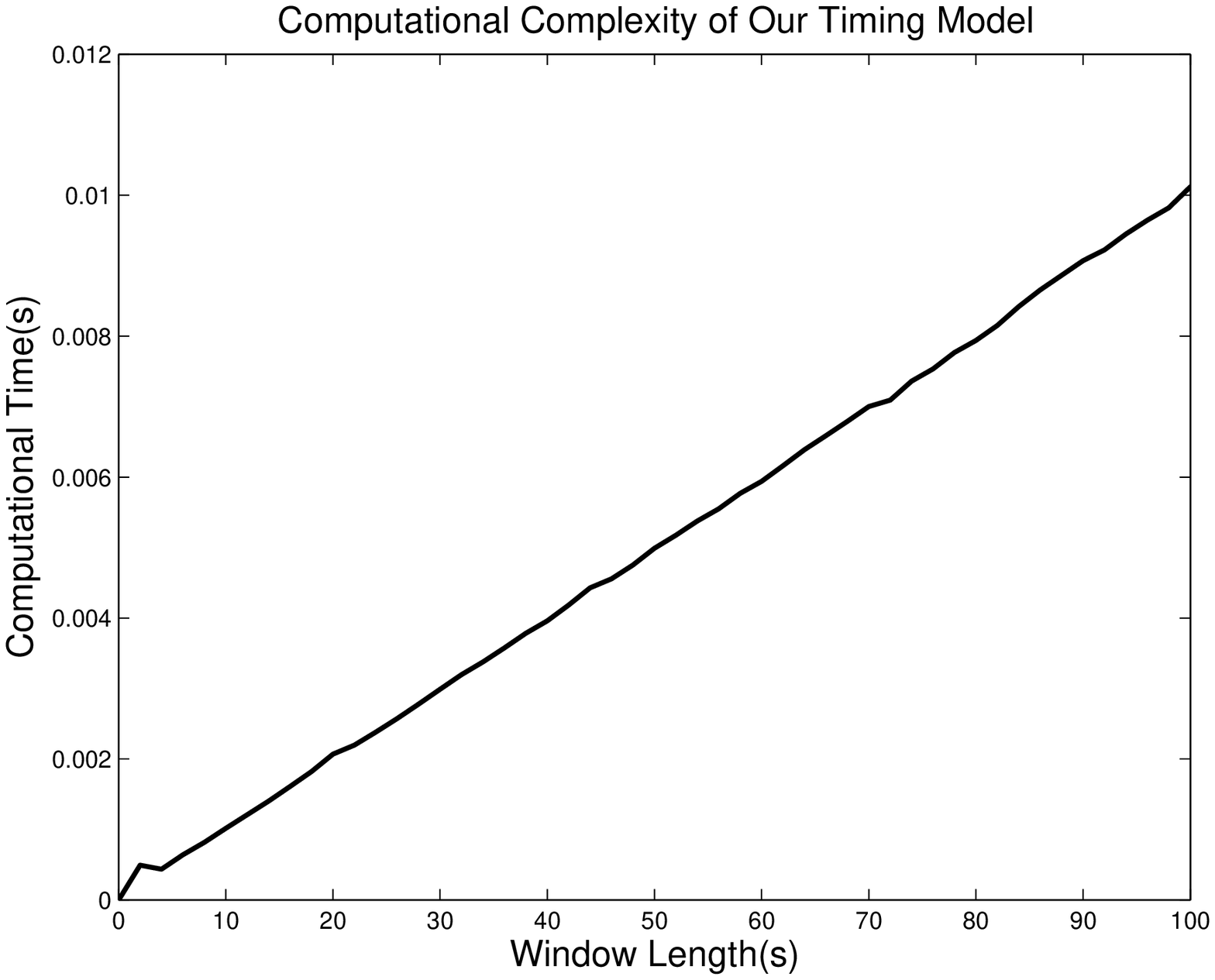}
  \label{fig:model_complexity}
 }
\subfigure[Truetime Simulation]
{\hspace{-1.5mm}\includegraphics[height=44mm, width=60mm]{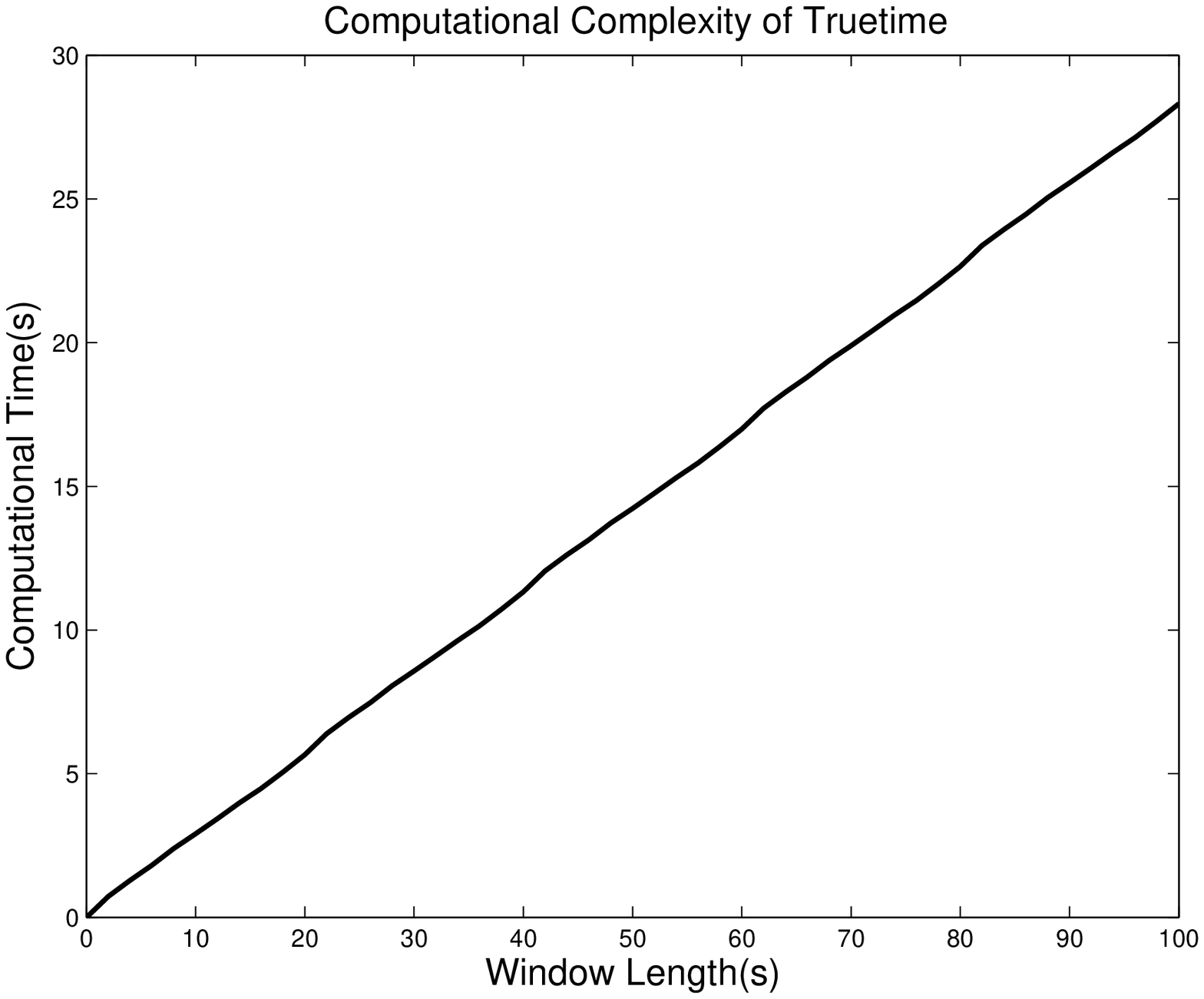}
  \label{fig:truetime_complexity}
 }
 \caption{Comparing the time needed to simulating scheduled behaviors.} \label{fig:ComputationalCost}
\end{figure}

\subsection{MPC Performance}
We demonstrate the performance of the MPC using the hybrid timing model for a CAN-based control system operating in dynamic, uncertain environment.  Suppose the messages on the CAN are changed at runtime. We consider two types of messages adjustments on the CAN within the time interval $[1, 1.5]$s. One is the adjustment of the message period as
\begin{equation}
\left[\,\hspace{0.7mm}T_1(t), \hspace{0.8mm}T_2(t), \hspace{0.8mm}T_3(t)\hspace{0.3mm}\,\right]=[20, 40, 50]\;{\rm ms}
\end{equation}
The other type of adjustments is the activation of two sporadic messages on the CAN, which have the following characteristics
\begin{align}
\left[\hspace{0mm}T_4(t), I_4^1(t), \hspace{0mm}C_4^1(t), \hspace{0mm}I_4^2(t),\hspace{0mm}C_4^2(t)  \hspace{0mm}\right]\hspace{-0.6mm}=\hspace{-0.6mm}[40, 0.2, 1, 0, 0]\,{\rm ms}\cr
\left[\hspace{0mm}T_5(t), I_5^1(t), \hspace{0mm}C_5^1(t), \hspace{0mm}I_5^2(t),\hspace{0mm}C_5^2(t)  \hspace{0mm}\right]\hspace{-0.6mm}=\hspace{-0.6mm}[60, 0.2, 1, 0, 0]\,{\rm ms}
\end{align}
The sporadic messages are assigned unique identifier field such that $P_5[k]<P_4[k]<P_1^{1}[k]$.  Note that since these adjustments happen at runtime, their characteristics are not available at the off-line
design stage. It is then expected that the timing of the message chains will be disturbed and the controller performance will be affected.

We compare two different approaches of designing MPC for the CAN-based control system. The two approaches differ in their way of predicting $\delta_n[k]$. In the first approach, the delay $\delta_n[k]$ is predicted off-line through the worst-case analysis discussed in \cite{Tindell,Tindell1995, CanRevisit}. In the second approach, the delay $\delta_n[k]$ is predicted online through the hybrid timing model. Figure 6 shows the MPC performance of three feedback control loops under the above two different approaches. The solid line represents the plant output $y_n(t)$ and the dashed line represents the reference trajectory $\gamma_n(t)$. The left plots are results of the first approach that uses the worst-case response time and the right plots are results of the second approach that uses the hybrid timing model. It is obvious that the second approach (right plots) gives better performance than the first approach (left plots). This is because in the second approach, delays are predicted online using the hybrid timing model of the CAN, which can accurately predict delay and dynamically compensate for the delay.

Also, it is worth mentioning that even in the first approach(left plots), MPC performance in the first feedback control loop is better than the other two loops. This is because the messages in the first feedback control loop are assigned the highest priorities among all messages on the CAN. Therefore, the difference between the actual delay and the worst-case response time is small in the first feedback control loop . Using even the worst-case response time for MPC design can still give out the acceptable performance for the first feedback control loop. However, such difference in the second and third feedback control loop will increases, which leads to the degraded MPC performance.

\begin{figure}[tp]
  \label{fig:twoapproaches}
\subfigure[The first feedback control loop]
{\hspace{10mm}\includegraphics[trim = 15mm 0mm 8mm 7mm, clip,height=30mm, width=43mm]{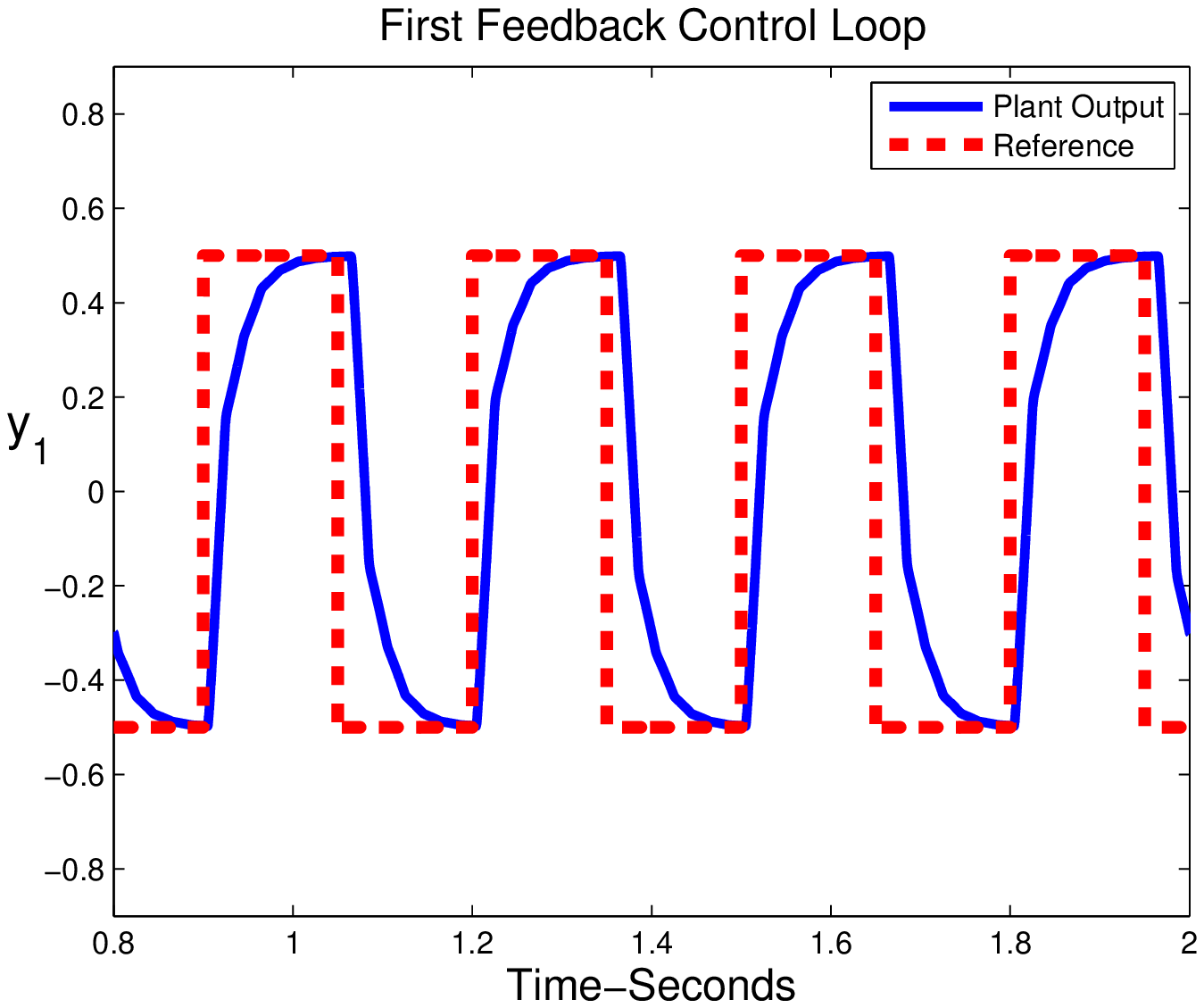}
\hspace{6mm}\includegraphics[trim = 15mm 0mm 8mm 7mm, clip,height=30mm, width=43mm]{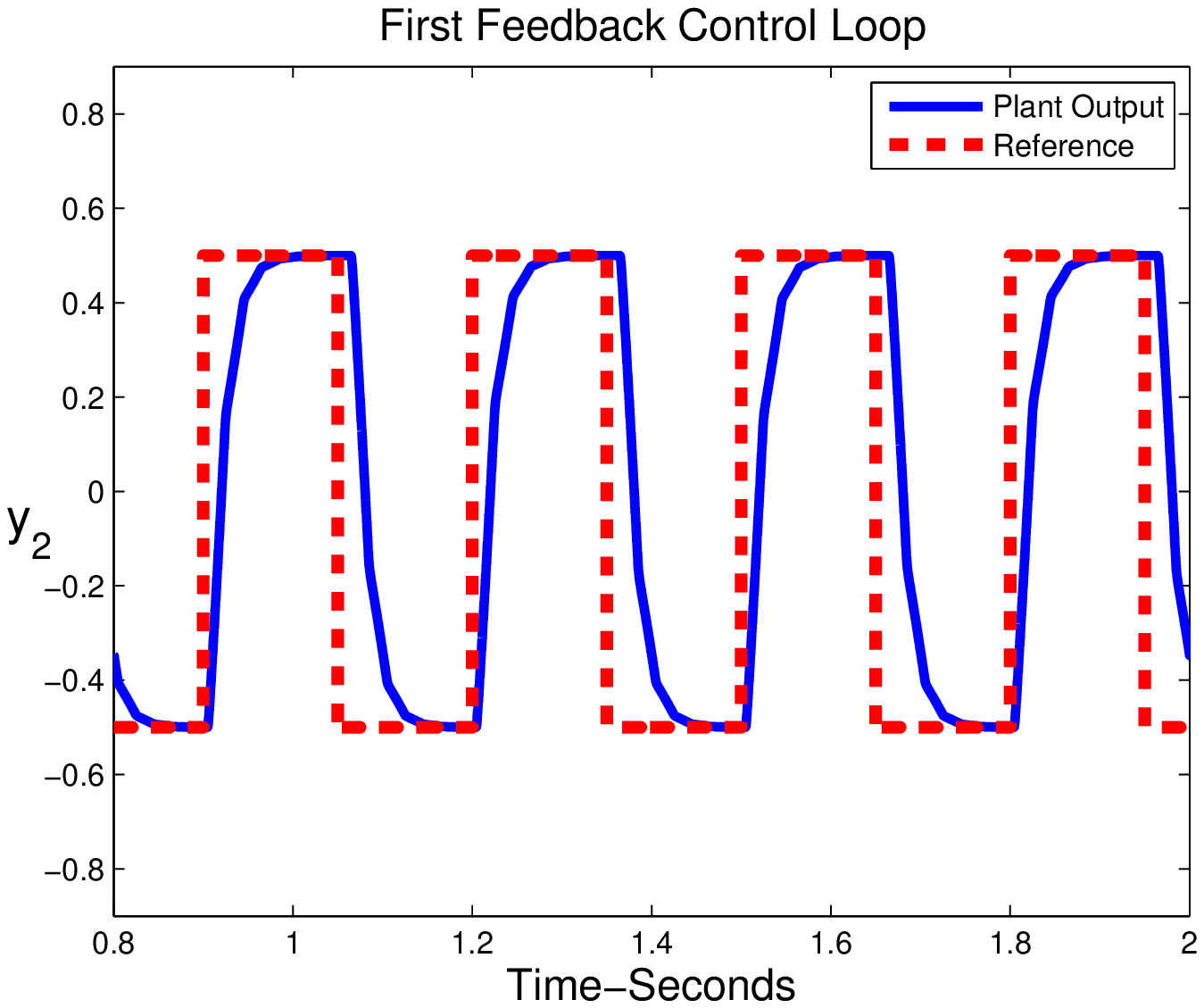}
\vspace{5mm}
 }
\subfigure[The second feedback control loop]
{\hspace{10mm}\includegraphics[trim = 15mm 0mm 8mm 7mm, clip,height=30mm, width=43mm]{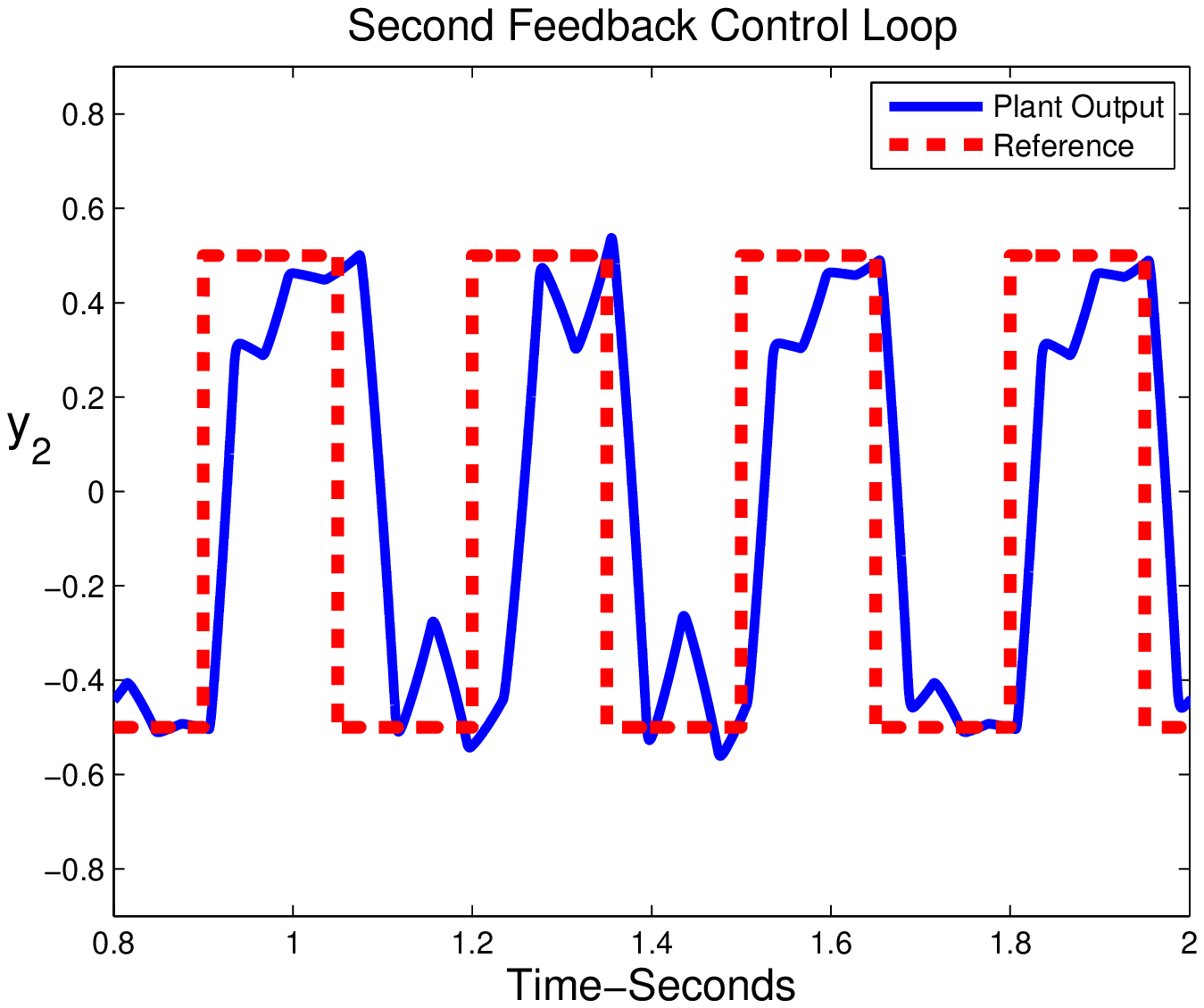}
\hspace{6mm}\includegraphics[trim = 15mm 0mm 8mm 7mm, clip,height=30mm, width=43mm]{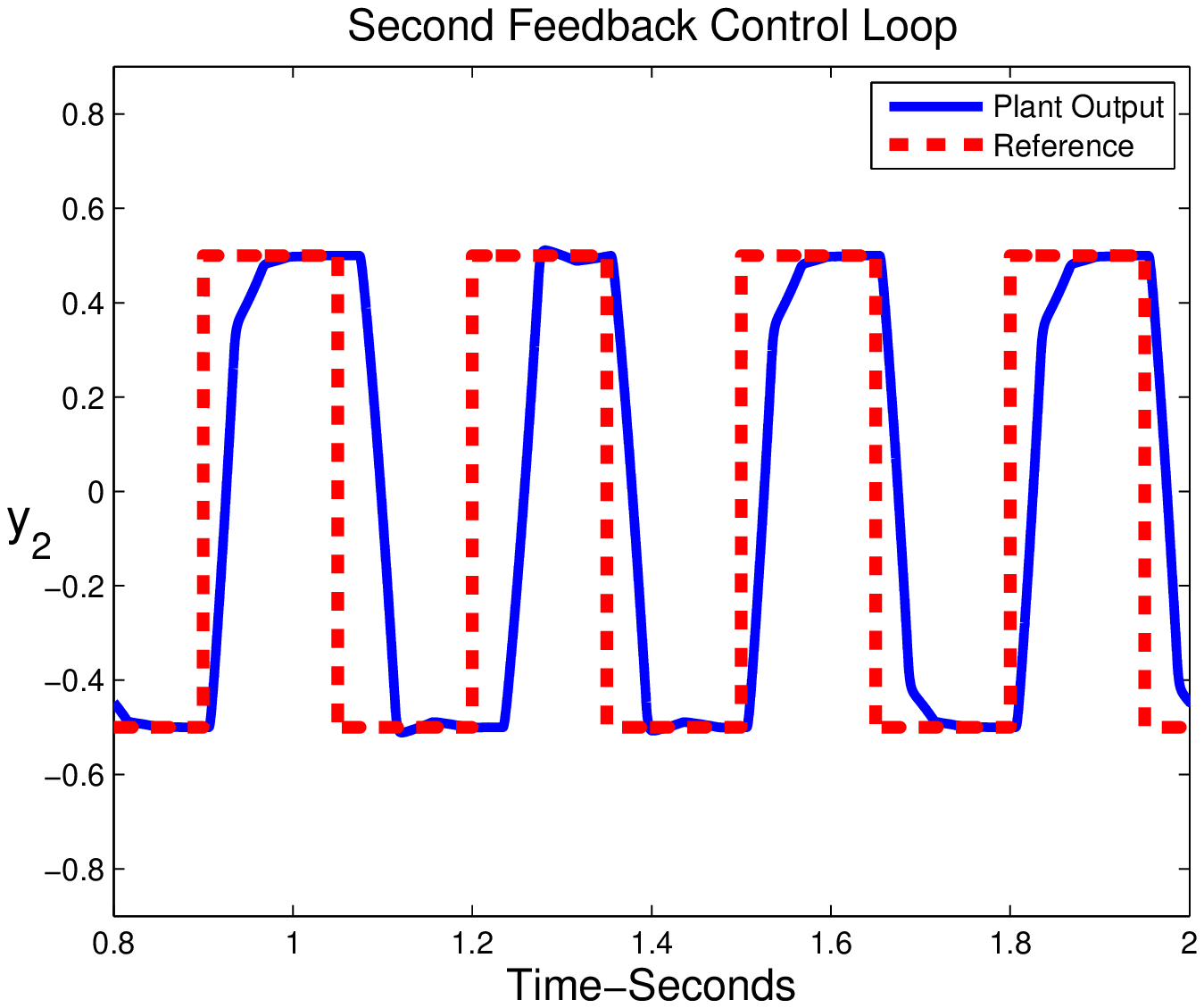}
\vspace{5mm}
 }
 \subfigure[The third feedback control loop]
{\hspace{10mm}\includegraphics[trim = 15mm 0mm 8mm 7mm, clip,height=30mm, width=43mm]{y3.eps}
\hspace{6mm}\includegraphics[trim = 15mm 0mm 8mm 7mm, clip,height=30mm, width=43mm]{y3_online.eps}
\vspace{5mm}
 }
 \caption{MPC performance of three feedback control loops, under two design approaches} \label{fig:MPCperformance}
\end{figure}

\section{Conclusion and Future Work}
 The main contribution of this paper is a hybrid timing model for messages scheduled on the CAN. We have shown that such timing model enables a model predictive control approach on the CAN.
 It also provides convenient ways to check for schedulability of messages. This model may be used for co-design of scheduling and MPC for real-time embedded systems on the CAN. Moreover, the timing model is a generic mathematical model that can be extended to many applications\cite{Wang2013,Wang2015,WangCPS2015, Shi2016}. Our simulations show that using the hybrid timing model for MPC can achieve improved performance than using worst case timing. Our future work will extend this hybrid timing model to other real-time communication networks that use message priorities for arbitration,  for example, the dynamic segment of FlexRay \cite{Pop2008}.



\end{document}